\documentclass[11pt,letterpaper]{article}
\pdfoutput=1

\usepackage{style}
\usepackage{shortcuts}

\def\eps{\epsilon}

\title{New Graph Decompositions and Combinatorial Boolean Matrix Multiplication Algorithms}
\author{Amir Abboud\thanks{Weizmann Institute of Science. This work is part of the project CONJEXITY that has received funding from the European Research Council (ERC) under the European Union's Horizon Europe research and innovation programme (grant agreement No.~101078482). Supported by an Alon scholarship and a research grant from the Center for New Scientists at the Weizmann Institute of Science. Email: \href{mailto:amir.abboud@weizmann.ac.il}{\nolinkurl{amir.abboud@weizmann.ac.il}}}
    \and Nick Fischer\thanks{Weizmann Institute of Science. Supported by the project CONJEXITY as above. Email: \href{mailto:nick.fischer@weizmann.ac.il}{\nolinkurl{nick.fischer@weizmann.ac.il}}}
    \and Zander Kelley\thanks{Department of Computer Science, University of Illinois at Urbana-Champaign. Supported by NSF CAREER award 2047310. Email: \href{mailto:awk2@illinois.edu}{\nolinkurl{awk2@illinois.edu}}}
    \and Shachar Lovett\thanks{
Department of Computer Science and Engineering, University of California, San Diego. Supported by NSF DMS award 1953928, NSF CCF award 2006443, and a Simons investigator award. Email: \href{mailto:slovett@ucsd.edu}{\nolinkurl{slovett@ucsd.edu}}}
    \and Raghu Meka\thanks{Department of Computer Science, University of California, Los Angeles. Supported by NSF AF 2007682 and NSF Collaborative Research Award 2217033. Email: \href{mailto:raghum@cs.ucla.edu}{\nolinkurl{raghum@cs.ucla.edu}}}}
\date{}

\begin{document}

\maketitle
\begin{abstract}
\noindent
We revisit the fundamental Boolean Matrix Multiplication (BMM) problem. 
With the invention of algebraic fast matrix multiplication over 50 years ago, it also became known that BMM can be solved in truly subcubic $O(n^{\omega})$ time, where $\omega<3$; much work has gone into bringing $\omega$ closer to $2$. 
Since then, a parallel line of work has sought comparably fast \emph{combinatorial} algorithms but with limited success.
The na\"{i}ve $O(n^3)$-time algorithm was initially improved by a $\log^2{n}$ factor [Arlazarov~\emph{et al.};~RAS'70], then by $\log^{2.25}{n}$ [Bansal and Williams; FOCS'09], then by $\log^3{n}$ [Chan;~SODA'15], and finally by $\log^4{n}$ [Yu;~ICALP'15].

We design a combinatorial algorithm for BMM running in time $n^3 / 2^{\Omega(\sqrt[7]{\log n})}$---a speed-up over cubic time that is stronger than any poly-log factor. 
This comes tantalizingly close to refuting the conjecture from the 90s that truly subcubic combinatorial algorithms for BMM are impossible.
This popular conjecture is the basis for dozens of fine-grained hardness results.

Our main technical contribution is a new \emph{regularity decomposition} theorem for Boolean matrices (or equivalently, bipartite graphs) under a notion of regularity that was recently introduced and analyzed analytically in the context of communication complexity [Kelley, Lovett, Meka; arXiv'23], and is related to a similar notion from the recent work on $3$-term arithmetic progression free sets [Kelley, Meka;~FOCS'23].
\end{abstract}

\setcounter{page}{0}
\thispagestyle{empty}
\clearpage

% !TEX root = ../paper.tex
\section{Introduction}

Boolean Matrix Multiplication (BMM) is one of the most basic and fundamental combinatorial problems. 
It can be solved in $O(n^{\omega})$ time, where $2 \leq \omega < 2.3716$~\cite{DuanWZ23,VassilevskaWilliamsXXZ23} is the exponent of (integer) matrix multiplication. %, by observing that it is a special case.
The algebraic technique underlying Strassen's~\cite{Strassen69} and all subsequent ``fast matrix multiplication'' algorithms have several limitations (discussed in Section~\ref{sec:discussion}) related to \emph{generalizability}, \emph{elegance}, and \emph{practical efficiency}. 
Therefore, in addition to the line of work trying to enhance this algebraic machinery aiming to reach $\omega=2$, a parallel line of work aims to match (or improve) these subcubic bounds with different combinatorial techniques.

The first result in this direction is the ``Four-Russians'' algorithm by Arlazarov, Dinic, Kronrod, and Farad\v{z}ev~\cite{ArlazarovDKF70} that achieves $o(n^3)$ complexity by precomputing the answers to small sub-instances.
This approach can give $O(n^3/\log^2 n)$ time, but nothing faster~\cite{Angluin76}.
After many decades, Bansal and Williams~\cite{BansalW12} used \emph{regularity lemmas} to gain an additional $\log^{0.25}{n}$ factor speed-up.
The usefulness of graph regularity techniques was undermined a few years later when Chan~\cite{Chan15} gave a better, $O(n^3/\log^{3}n)$ bound only using simple divide-and-conquer.
Yu~\cite{Yu18} optimized the divide-and-conquer method to achieve an  $O(n^3/\log^{4}n)$ bound that stood since 2015.\footnote{All running times in this paragraph are up to log-log-factors.} 
%\nick{All these running times are up to loglog factors}\shachar{we can add this comment as a footnote}

A pessimistic conjecture that has been popular since the 90s~\cite{Satta94,Lee02} states that truly subcubic running times are impossible for combinatorial algorithms; that is, one may be able to shave some more logarithmic factors, but we cannot reach $O(n^{3-\eps})$ time with $\eps>0$.

\begin{conjecture}[Combinatorial BMM] \label{conj:bmm}
There is no combinatorial algorithm for BMM running in time $O(n^{3-\eps})$, for any $\eps>0$.
\end{conjecture}

This conjecture has served as the basis for many conditional lower bounds.
Refuting it would re-open the quest for polynomially faster combinatorial algorithms for a host of fundamental discrete problems~\cite{RodittyZ11,WilliamsW18,AbboudW14,AbboudWY18,CliffordGLS18,ChanR020,CaselS21,Chang19,BringmannGL17,LincolnWW18,BringmannW17,DahlgaardKS17,AbboudBBK17,BringmannGMW20,DalirrooyfardVW21,AbboudGIKPTUW19,BringmannFK19,BergamaschiHGWW21,JinX22,HuangLSW23}\footnote{Some of these references assume the stronger Combinatorial $k$-Clique Conjecture, that would also be refuted if \cref{conj:bmm} is false.}.%todo: i didn't add all the BMM references yet, so either add them, or point to the survey by virginia. %Nick: I added all the references from her survey

\paragraph{Main Result.} In this paper, we prove a new regularity decomposition theorem that leads to a \emph{quasi-polynomial}\footnote{Recall that bounds of the form $2^{(\log n)^c}$ are called quasi-polynomial. In our case, $c<1$.} saving for BMM, combinatorially, coming much closer than before to refuting \cref{conj:bmm}.

\begin{theorem}[Combinatorial BMM] \label{thm:bmm}
There is a deterministic combinatorial algorithm computing the Boolean product of two $n \times n$ matrices in time~\smash{$n^3 / 2^{\Omega(\sqrt[7]{\log n})}$}.
\end{theorem}

We immediately get a similar combinatorial  \emph{super-poly-logarithmic} saving for many other problems (that can be reduced to BMM).
The list includes central problems in their domains such as \emph{context-free grammar parsing} from formal languages~\cite{Valiant75}, computing the \emph{transitive closure} of a directed graph~\cite{FischerM71,Munro71}, \emph{join-project queries} from databases~\cite{AmossenP09,GuchtWWZ15}, parameterized problems such as \emph{$k$-clique}~\cite{NesetrilP85} and \emph{$k$-dominating-set}~\cite{EisenbrandG04}, and various matrix product problems~\cite{Matousek91,CzumajKL07,VassilevskaWY07}.

Our (and most previous) results are obtained by designing algorithms for the simpler \emph{triangle detection} problem and then using a well-known subcubic equivalence with BMM~\cite{WilliamsW18}.  Therefore, we focus on triangle detection below.

\paragraph{A New Regularity Decomposition Theorem}
In abstract terms, a graph is \emph{regular} if it behaves somewhat pseudo-randomly, and a \emph{regularity decomposition} theorem states that any graph can be decomposed into a ``small'' number of regular subgraphs. 
% A regularity decomposition theorem says that any graph can be decomposed into a ``small'' number of random-like or ``regular'' graphs.
Such results are interesting mathematically because they say that any graph can be simplified dramatically, and also algorithmically because they let us reduce a problem from arbitrary to (the often easier) random-like graphs.
% In principle, such a technique could give us an $n^{2\frac{1}{3}}$ algorithm for triangle detection and $n^{2.5}$ for BMM.\shachar{I don't understand this sentence. Why is this possible in principle? why these bounds? this might be confusing to non experts}
%\todo{Comment: BMM on random instances is in time $\Order(n^{2.5})$. Here or somewhere else?}\shachar{I don't think it fits the narrative here}
The possibility and efficiency of such results depend on the precise notion of regularity; generally, there is a trade-off between strength (i.e.\ how close ``regular'' is to random) and efficiency (i.e.\ the number of subgraphs in the decomposition).

For example, the celebrated Szemerédi's Regularity Lemma~\cite{Szemeredi75} yields a decomposition into subgraphs with very strong pseudo-random properties, but to achieve meaningful results for graphs with density $\delta$, the number of parts inherently scales as a tower function of height $\poly(1/\delta)$~\cite{Gowers97}.\footnote{I.e., \smash{$f(\delta) \leq 2^{2^{2^{\dots}}}$} where the tower has height $\poly(\delta^{-1})$.} A comparable but weaker notion of regularity due to Frieze and Kannan~\cite{FriezeK99} admits decompositions into fewer, but still exponentially many parts (specifically, \smash{$2^{\Order(\delta^{-2})}$}). At the other end of the spectrum, \emph{expander decompositions} use a much weaker notion of pseudo-randomness, but significantly gain in efficiency.

% For example, the celebrated \emph{Szemer\'{e}di-regularity} is a very strong notion, and the number of subgraphs is infamously large (tower-type), while the powerful algorithmic technique of \emph{expander-decompositions} uses a much weaker notion and gains in efficiency.

Considering a specific problem, e.g.\ triangle detection, the challenge is finding a sweet spot in which the regularity notion is strong enough to make the problem algorithmically easy, yet weak enough to make the decomposition efficient.  
Unfortunately, expanders are too weak~\cite{AbboudW23}. Based on Szemerédi-regularity and Frieze-Kannan-regularity, Bansal and Williams~\cite{BansalW12} indeed scored non-trivial algorithmic improvements for Boolean Matrix Multiplication: A $(\log^*(n))^{\Omega(1)}$-shave based on Szemerédi's Regularity Lemma,\footnote{In fact, their algorithm is based on the Triangle Removal Lemma which can be proven via Szemerédi's Regularity Lemma, but could in principle (and, in fact, does) admit better quantitative bounds~\cite{Fox11}.} and a $(\log n)^{1/4}$-shave based on Frieze-Kannan regularity. In both cases, however, due to the excessive number of pieces in the regularity decomposition, it seems hopeless to go beyond log-shaves.

In this paper we employ a regularity notion called \emph{grid regularity} that was recently introduced by Kelley, Lovett and Meka~\cite{KelleyLM23} in the context of communication complexity and is based on similar results in the work of Kelley and Meka~\cite{KelleyM23} on $3$-term arithmetic progressions. This regularity notion is weak but still useful for triangle detection and BMM. Phrased in terms of matrices, the key takeaway from their work is that a single grid-regular matrix is not necessarily very random, but the \emph{product of two grid-regular matrices is very random} (see \cref{thm:regular-product}). Equivalently, in a 3-layered graph in which both edge sets are grid regular, the number of 2-paths from left to right behaves randomly.
Our main contribution is that we (1) establish a decomposition theorem for this notion of regularity into \emph{quasi-polynomially} many parts (specifically,~\smash{$2^{\Order((\log \delta^{-1})^7)}$}), and~(2)~provide an efficient deterministic algorithm to compute this decomposition.

% The main technical result of this paper is an efficient regularity decomposition based on a new notion of regularity that is weak but still useful for triangle detection and BMM: The key point is that, while a single regular matrix is not very random, but the \emph{product of two regular matrices is very random}. 
% This notion, and the fact that it is useful for triangles, follow from the work of Kelley, Lovett and Meka~\cite{KelleyLM23} on communication complexity and is based on similar results in the celebrated result of Kelley and Meky~\cite{KelleyM23} on $3$-term arithmetic progressions. Our main contribution is in showing that it allows for an efficient decomposition.

\subsection{On Combinatorial Algorithms}
\label{sec:discussion}

Despite the large number of papers on combinatorial algorithms for BMM, there is currently no satisfying and precise definition for this notion. 
The main reason for this, we believe, is that there are multiple strong motivations for seeking combinatorial algorithms (discussed below) that are not necessarily consistent with one another.
A simple algorithm need not be practical, or vice versa, and an algorithm that generalizes to one setting may not generalize to another.
Therefore, one can either focus on precise definitions that are limited to one motivation or embrace a more inclusive but loose definition.
Many examples of the former approach exist (see below), and they give rise to interesting research questions.
The latter, however, is more popular in the community: we currently have no truly subcubic algorithm for BMM other than Strassen's algorithm (and its successors) so we should first seek to break \cref{conj:bmm} by \emph{any other technique} and hope that \emph{at least one} of the motivations gets satisfied. 

To help shed light on this, let us review the limitations of the existing algebraic technique that motivate us to seek other algorithms.
The first two may be more well-known, but the third is more pressing for fine-grained complexity and algorithm design.
Along the way, we discuss to what extent our algorithm satisfies each consideration.

\begin{itemize}
\item \textbf{Simplicity:} Strassen's algorithm (and even more so for its successors) exploits cancellations using formulas that may be considered unintuitive; consequently, the values manipulated at intermediate stages of the computation are quite uninterpretable. 
%Even if we were to reach $\omega=2$ with this technique, it is likely that the algorithm will be terribly complicated. 
Instead, one may hope for techniques that are simpler and more interpretable; this is probably the historical reason for the name ``combinatorial''. Some works have proposed precise definitions along these lines, e.g.\ for solving triangle detection an algorithm can only generate sets by basic operations on the neighborhoods of nodes, and strong lower bounds exist~\cite{Angluin76,DasKS18}. 
Unfortunately, the current definitions are not flexible enough to capture the regularity decomposition techniques of Bansal-Williams~\cite{BansalW12} and of our paper, nor even the simple divide-and-conquer of Chan~\cite{Chan15} and Yu~\cite{Yu18}, even though these techniques are widely-considered ``combinatorial'' (even by the authors proposing the definitions).
In particular, our algorithm decomposes the graph by repeatedly finding and removing irregular pieces (i.e., that contain too many bicliques) and then uses a brute-force algorithm on the (sparse) parts.

\item \textbf{Practical Efficiency}: Here, one should make a distinction between Strassen's $n^{2.81}$ algorithm, and its successors that have reduced $\omega$ much further. The latter algorithms are considered ``galactic'' and the interest in them is mostly theoretical. On the other hand, Strassen's algorithm has been used in practice, but the gains are limited; some reasons include its bad locality (generating too many cache misses) and the need to manipulate large numbers.
For many decades, researchers have sought techniques that are more efficient in practice \emph{and also} have worst-case guarantees.\footnote{Note that the latter requirement is crucial because very fast matrix multiplication algorithms do exist in practice, but are due to highly optimized implementations of the cubic time solution and to special-purpose hardware whose sole purpose is to perform these implementations. Needless to say, a real algorithmic speedup is still desirable.} The lack of success in finding such algorithms partly motivated \cref{conj:bmm} in the 90s~\cite{Satta94,Lee02,AingworthCIM99}.
A definition of this notion has to be empirical, and determining if our algorithm satisfies it requires experiments. Regularity decompositions are infamously impractical, but our underlying paradigm of decomposing the input into pseudo-random parts has the potential to be practical.
%(perhaps if the $n^{2.5}$ limit is reached). \nick{I don't understand why $n^{2.5}$ is a limit}

\item \textbf{Generalizability:} Much of the interest in BMM and triangle detection is because they are simplified special cases of more difficult, important problems. 
A technique that only solves the special case is not satisfying.
In the following, we give four examples of such problems where (1) for well-established reasons, Strassen's technique does not give a truly subcubic algorithm, and (2) a truly subcubic algorithm that does generalize would be groundbreaking as it would refute a popular conjecture in fine-grained complexity that has nothing to do with ``combinatorial algorithms'', i.e.\ we do not know how to refute it with any algorithms.
In each case, a precise definition of ``combinatorial'' in the sense of generalizing to that particular setting can readily be made: it is an algorithm that solves the corresponding problem.

\begin{itemize}
\item To refute the famous All-Pairs Shortest-Paths (APSP) Conjecture, it is enough (in fact, equivalent) to solve \emph{weighted} generalizations of BMM and triangle detection in truly subcubic time: $(min,+)$-matrix-multiplication and Negative-Triangle~\cite{WilliamsW18}.
The issue with Strassen's technique is that it exploits cancellations by subtracting numbers, and $\min$ does not have an inverse.

\item To refute the Online Matrix-Vector (OMV) Conjecture~\cite{HenzingerKNS15}, we need to solve BMM in an \emph{online} setting in which the columns of the second matrix arrive one by one and, at each step, we must output the answer before seeing the next column, in a total time that is truly subcubic.
The issue with Strassen's algorithm is that its formulas depend on later columns.

\item To refute the Hyper-Clique Conjecture~\cite{LincolnWW18}, we need a generalization to \emph{hypergraphs} that lets us detect a $4$-clique in a $3$-uniform hypergraph. 
A technique entirely different from Strassen's is needed because formulas that reduce the number of multiplications provably do not exist---the border rank of the corresponding tensor matches the trivial upper bound~\cite{LincolnWW18}.

\item To refute the famous $3$-SUM Conjecture~\cite{Patrascu10,KopelowitzPP16} (and also the APSP Conjecture~\cite{VassilevskaWilliamsX20}), it is enough to obtain a generalization to the (witness) \emph{reporting} setting. In particular, we would like an algorithm that preprocesses a graph in truly subcubic time and can then enumerate all triangles with constant delay.
Unfortunately, the witness information is lost under the cancellations that are exploited in Strassen's algorithm.

\end{itemize}

Does our new algorithm and, more generally, the regularity decompositions technique generalize to these four settings?
For the first three, it is unclear and left for future research (there was much less incentive to do it before this work since it was outperformed by divide-and-conquer).\footnote{For the first~\cite{Williams18} and second~\cite{LarsenW17} settings, breakthrough works have already managed to shave a quasi-polynomial factor by fancy reductions to matrix multiplication (and then using the algebraic methods); nonetheless, achieving such improvements with our combinatorial methods would be interesting. For the third setting, it is a big open question.}
For the reporting setting, Abboud, Fischer, and Shechter~\cite{AbboudFS23} recently observed that the ideas in the Bansal-Williams algorithm can give a triangle enumeration algorithm with $n^3/\log^{2.25}n$ preprocessing and constant delay, which is the state-of-the-art even when using algebraic techniques.
Building on our new decomposition theorem, we give an improved bound, demonstrating that our techniques are useful beyond ``combinatorial'' algorithms.

\begin{restatable}[Triangle Enumeration Algorithm]{theorem}{thmtriangleenum} \label{thm:triangle-enumeration}
There is a deterministic algorithm that preprocesses a given graph in time~\smash{$n^3 / (\log n)^6 \cdot (\log\log n)^{\Order(1)}$} and then enumerates all triangles with constant delay.
\end{restatable}

It should appear strange that we shave only a $\log^6{n}$ factor and not a super-poly-log. The reason for this is that $O(n^3/\log^6{n})$ is the best we know how to achieve (using any technique) for triangle enumeration \emph{even on random graphs} (see \cref{sec:triangle-enumeration:sec:overview}).
Thus, it is a natural limit for any technique (like ours) based on a reduction to random-like instances. 
Moreover, due to the reduction from $3$-SUM to triangle enumeration~\cite{Patrascu10,KopelowitzPP16}, shaving any additional $\log^{\eps}{n}$ factor over our bound would improve on the longstanding upper bound for (integer) $3$-SUM~\cite{BaranDP08} (see \cref{sec:triangle-enumeration:sec:lower-bound}).
%\raghu{Would a $(\log n)^\epsilon$ savings for enumeration imply a truly sub-quadratic algorithm for 3-SUM? If so, it would be good to make it explicit here.} \amir{no, it would only improve the $n^2/\log^2n$ bound by a $\log^{\eps'}n$ factor.}
\end{itemize}

\section{Preliminaries} \label{sec:preliminaries}
We write $[n] = \set{1, \dots, n}$ and $\poly(n) = n^{\Order(1)}$. Occasionally, we write $a = b \pm \epsilon$ to express that $|a - b| \leq \epsilon$.

\subsection{Machine Model}
We assume the standard Word RAM model with word size $\Theta(\log n)$ (where $n$ is the input size). Since, for most of our algorithmic results, additional log-factors in the running times would not matter, the choice of the machine model is not crucial. Only in \cref{sec:triangle-enumeration}, when we care about log-factors, this choice matters.

\subsection{Graphs and Matrices}
We typically denote sets (of nodes) by $X, Y, Z$ and matrices by $A, B, C$. Moreover, we typically view binary matrices $A \in \set{0, 1}^{X \times Y}$ as bipartite graphs on the node sets $X, Y$, where an edge $(x, y)$ is present if and only if $A(x, y) = 1$.

Let \smash{$A \in \Real_{\geq 0}^{X \times Y}$} and \smash{$B \in \Real_{\geq 0}^{Y \times Z}$} be matrices. We denote by $A B$ their standard matrix product, and by $A \circ B$ a scaled matrix product defined by
\begin{equation*}
    (A \circ B)(x, z) = \Ex_{y \in Y} A(x, y) B(y, z) = \frac{1}{|Y|} \cdot \sum_{y \in Y} A(x, y) B(y, z).
\end{equation*}
Following the bipartite graph analogy, for sets $X' \subseteq X$ and~\makebox{$Y' \subseteq Y$}, we let \smash{$A[X', Y'] \in \Real_{\geq 0}^{X \times Y}$} denote the submatrix restricted to the rows in $X'$ and the columns in $Y'$---that is, the subgraph induced by $X' \cup Y'$. Let $A^T$ denote the transpose of $A$---that is, the subgraph obtained by exchanging the sides $X$ and $Y$. We call
\begin{equation*}
    \Ex[A] = \Ex_{\substack{x \in X\\y \in Y}} A(x, y) = \frac{1}{|X| \, |Y|} \cdot \sum_{\substack{x \in X\\y \in Y}} A(x, y)
\end{equation*}
the \emph{density} of $A$. For nodes $x \in X$ and $y \in Y$, we define their \emph{(relative) degrees} as
\begin{align*}
    \deg_A(x) &= \Ex_{y \in Y} A(x, y) = \frac{1}{|Y|} \cdot \sum_{y \in Y} A(x, y), \\
    \deg_A(y) &= \Ex_{x \in X} A(x, y) = \frac{1}{|X|} \cdot \sum_{x \in X} A(x, y).
\end{align*}
We say that $A$ is \emph{$\epsilon$-left-min-degree} or simply \emph{$\epsilon$-min-degree} if $\min_{x \in X} \deg_A(x) \geq (1 - \epsilon) \Ex[A]$. (The symmetric notion of \emph{$\epsilon$-right-min-degree} is never used in the paper.) Moreover, for~\makebox{$\alpha, \epsilon, \delta \geq 0$}, we say that $A$ is \emph{$(\alpha, \epsilon, \delta)$-uniform} if
\begin{equation*}
    \Pr_{(x, y) \in X \times Y}[\,(1 - \epsilon) \alpha \leq A(x, y) \leq (1 + \epsilon) \alpha\,] \geq 1 - \delta.
\end{equation*}
% \shachar{this will need to change to $(\alpha,\epsilon,\delta)$, where $\alpha$ replaces $\Ex[A]$. Also, we call it near-uniform in the NOF paper. It will be good to use the same terminology in both papers. What do you like better, uniform or near-uniform?} \nick{I like uniform better, but I don't mind changing it. :)}

\subsection{Grid Regularity}
Recall that we abstractly consider a graph \emph{regular} if it behaves somewhat pseudo-randomly.\footnote{This concept is \emph{not} related to degree-regularity (where each node in the graph has the same number of neighbors).} In this paper we employ the following formal notion of regularity, defined via the \emph{``grid norm''} of a matrix. Specifically, for a matrix with non-negative entries~\smash{$A \in \Real_{\geq 0}^{X \times Y}$} and integers $k, \ell \geq 1$, we define its $(k,\ell)$-grid norm as
\begin{equation*}
    \norm{A}_{U(k, \ell)} = \parens*{\Ex_{\substack{x_1, \dots, x_k \in X\\y_1, \dots, y_\ell \in Y}} \prod_{\substack{i \in [k]\\j \in [\ell]}} A(x_i, y_j)}^{\frac{1}{k \ell}};
\end{equation*}
note that equivalently
\begin{equation*}
    \norm{A}_{U(k, \ell)}^{k \ell } = \Ex_{x_1, \dots, x_k \in X} \parens*{\Ex_{y \in Y} \prod_{i \in [k]} A(x_i, y)}^\ell = \Ex_{y_1, \dots, y_\ell \in Y} \parens*{\Ex_{x \in X} \prod_{i \in [\ell]} A(x, y_j)}^k.
\end{equation*}
Strictly speaking, $\norm{\,\cdot\,}_{U(k, \ell)}$ is not necessarily a norm, but we will nevertheless intuitively treat it as such.\footnote{It is however known that $\norm{\,\cdot\,}_{U(k, \ell)}$ is a semi-norm whenever $k$ and $\ell$ are even~\cite[Theorems~2.8 and~2.9]{Hatami2010}.} In combinatorial terms, the grid norm of a bipartite graph \smash{$A \in \set{0, 1}^{X \times Y}$} measures (up to normalization) the number of $(k, \ell)$-bicliques that occur as subgraphs of $A$ (including subgraphs in which some nodes of the biclique coincide).

Note that the grid norm $\norm{A}_{U(k, \ell)}$ ranges from $\Ex[A]$ to $1$, and thereby constitutes some measure of pseudo-randomness: On the one hand, purely random bipartite graphs (with edge density~$\Ex[A]$) have grid norm~\makebox{$\norm{A}_{U(k, \ell)} \approx \Ex[A]$}, whereas structured graphs (e.g., graphs with large induced subgraphs of increased density) often have larger grid norms. In this spirit, we say that $A$ is \emph{$(\epsilon, k, \ell)$-regular} if
\begin{equation*}
    \norm{A}_{U(k, \ell)} \leq (1 + \epsilon) \Ex[A].
\end{equation*}
For specific constant values of $k$ and $\ell$, grid norms have appeared in many previous mathematical works (e.g.,~\cite{Gowers01,Gowers06}), and also implicitly in recent algorithmic structure-to-randomness reductions~\cite{AbboudBKZ22,AbboudBF23,JinX23}.

\subsection{Kelley-Lovett-Meka's Structural Theorem}
What makes grid norms useful for us? In a recent result, Kelley, Lovett and Meka \cite{KelleyLM23} use analytical methods to obtain the following structural result, linking the regularity of two graphs $A$ and $B$ to their product matrix.

\begin{theorem}[Regular Matrices Have Uniform Products~{\cite[Lemma~4.8]{KelleyLM23}}] \label{thm:regular-product}
Let \smash{$A \in \Real_{\geq 0}^{X \times Y}$} and \smash{$B \in \Real_{\geq 0}^{Y \times Z}$}, let $\epsilon \in (0, \frac{1}{80})$, let $d \geq 2/\epsilon$ and assume that
\begin{enumerate}[label=(\alph*)]
    \item $A$ and $B^T$ are $(\epsilon, 2, d)$-regular, and
    \item $A$ and $B^T$ are $\epsilon$-min-degree.
\end{enumerate}
Then $A \circ B$ is $(\Ex[A] \Ex[B], 80\epsilon, 2^{-\epsilon d / 2})$-uniform.
\end{theorem}
% \shachar{need to revise using the new notion of uniformity with density $\Ex[A] \Ex[B]$}

% \todo{Add a half-page appe0ndix section that describes how to obtain this theorem from~\cite{KelleyLM23}?}

% \begin{theorem}[Regular Matrices Have Uniform Products~{\cite[Lemma~4.8]{KelleyLM23}}] \label{thm:regular-product}
% Let \smash{$A \in \Real_{\geq 0}^{X \times Y}$} and \smash{$B \in \Real_{\geq 0}^{Y \times Z}$}, let $\epsilon \in (0, \frac{1}{80})$, let $d \geq 2$ be an even integer and assume that
% \begin{enumerate}[label=(\alph*)]
%     \item $A$ and $B^T$ are $(\epsilon, 2, \ceil{d / \epsilon})$-regular, and
%     \item $A$ and $B^T$ are $\epsilon$-min-degree.
% \end{enumerate}
% Then $A B$ is $(80\epsilon, 2^{-d})$-uniform.
% \end{theorem}
% !TEX root = ../paper.tex
\section{Technical Overview} \label{sec:overview}
\cref{thm:regular-product} is the starting point
%\amir{note that people often skip the prelims section so it's better not to start the section with "this".}
Our goal in the following is to exploit this structural theorem algorithmically and derive an improved combinatorial algorithm for Boolean matrix multiplication. In this section we describe our key ideas.

\paragraph{Boolean Matrix Multiplication and Triangle Detection}
Recall that the Triangle Detection problem is to test whether a given undirected, tripartite graph $(X, Y, Z, A, B, C)$ (with vertex parts~$X, Y, Z$ and edge parts \smash{$A \in \set{0, 1}^{X \times Y}, B \in \set{0, 1}^{Y \times Z}, C \in \set{0, 1}^{X \times Z}$}) contains a \emph{triangle} (that is, a vertex triple $(x, y, z) \in (X, Y, Z)$ with $A(x, y) = B(y, z) = C(x, z) = 1$). While at first glance Triangle Detection appears to be a simpler problem than BMM (note that the output consists of a single bit versus $n^2$ bits), it is known since the early days of fine-grained complexity that both problems are, in fact, equivalent in terms of subcubic algorithms~\cite{WilliamsW18}. Specifically, if Triangle Detection can be solved in time $\Order(n^3 / f(n))$, then Boolean Matrix Multiplication is in time~\makebox{$\Order(n^3 / f(n^{1/3}))$}. This reduction is essentially loss-less for the quasi-polynomial speed-up that we aim for in this paper. Therefore, we focus on designing an efficient algorithm for Triangle Detection in the following exposition.

\paragraph{Triangle Detection on Regular Graphs}
We start by describing a dream scenario to understand how Kelley-Lovett-Meka's structural theorem~\cite{KelleyLM23} comes into play.

Our aim is to solve Triangle Detection in time $n^3 / 2^{\Omega(d)}$ for some parameter $d$. Moreover, let~\makebox{$\epsilon > 0$} be a small constant (say, $\epsilon = \frac{1}{160}$). For the dream scenario suppose that the edge parts $A$ and $B$ are regular in the sense of \cref{thm:regular-product}:
\begin{enumerate}[label=(\alph*)]
    \item $A$ and $B^T$ are $(\epsilon, 2, d)$-regular, and
    \item $A$ and $B^T$ are $\epsilon$-min-degree.
\end{enumerate}
Under these assumptions, \cref{thm:regular-product} yields that the scaled matrix product $A \circ B$ is $(\Ex[A] \Ex[B], 80\epsilon,\allowbreak 2^{-\epsilon d/2})$\-/uniform. Explicitly, for our choice of $\epsilon=\frac{1}{160}$, this means that
at least a $(1-2^{-\epsilon d/2})$-fraction of the entries of~$A \circ B$ fall in the range $[\frac{1}{2} \Ex[A] \Ex[B], \frac{3}{2} \Ex[A] \Ex[B]]$. In particular, the matrix $A \circ B$ (and thereby also $A B$) has zeros in at most a $2^{-\epsilon d/2}$-fraction of its entries (assuming that $A$ and $B$ are nonzero).
%\shachar{this is where the new thm is better, you don't need to assume $\Ex[AB]$ is nonzero, since it is replaced with $\Ex[A] \Ex[B]$ which we know is nonzero}

This puts us in a win-win situation: Either the matrix $C$ is sparse ($\Ex[C] \leq 2^{-\epsilon d/2}$), in which case we can detect a triangle in time $n^3 / 2^{\Omega(d)}$ (by enumerating all $n^2 / 2^{\Omega(d)}$ edges in $C$ and all remaining nodes~$y \in Y$). Or the matrix $C$ is dense ($\Ex[C] > 2^{-\epsilon d/2}$), and it follows from the uniformity that~$AB$ and~$C$ have a common nonzero entry. Note that this certifies that there is a triangle without the need to compute anything further.

\paragraph{Our Regularity Decomposition}
Of course, we cannot simply assume the dream scenario where~$A$ and~$B$ satisfy the regularity and min-degree conditions. Instead, we hope to decompose $A$ and $B$ into regular pieces in the same flavor as Szemerédi's or Frieze-Kannan's regularity lemmas. For grid regularity, unfortunately, such a decomposition theorem was not known.

One of our key contributions is such a decomposition theorem, see \cref{thm:2-path-decomposition}. 
%It is our key contribution that we provide the missing decomposition result, see \cref{thm:2-path-decomposition}. 
We emphasize that this theorem is novel even existentially (i.e., even without the extra requirement that the decomposition must be computed efficiently).

\begin{restatable}[$AB$-Decomposition]{theorem}{thmtwopathdecomposition} \label{thm:2-path-decomposition}
Let $A \in \set{0, 1}^{X \times Y}, B \in \set{0, 1}^{Y \times Z}$, let $\epsilon \in (0, 1)$ and~\makebox{$d \geq 1$}. There is an algorithm $\ABDecomposition(X, Y, Z, A, B, \epsilon, d)$ that computes a collection of tuples $\set{(X_k, Y_k, Z_k, A_k, B_k)}_{k=1}^K$, where $X_k \subseteq X$, $Y_k \subseteq Y$, $Z_k \subseteq Z$, $A_k \in \set{0, 1}^{X_k \times Y_k}$, \smash{$B_k \in \set{0, 1}^{Y_k \times Z_k}$} such that
\begin{enumerate}
    \item $A B = \sum_{k=1}^K A_k B_k$.
    \item For all $k \in [K]$:
    \begin{enumerate}[label=(\roman*)]
        \item $\Ex[A_k] \leq 2^{-d}$ or $\Ex[B_k] \leq 2^{-d}$, or
        \item $A_k$ and $B_k^T$ are both $(\epsilon, 2, d)$-regular and $\epsilon$-min-degree. 
    \end{enumerate}
    \item $\sum_{k=1}^K |X_k| \, |Y_k| \, |Z_k| \leq 2(d + 2)^2 \cdot |X| \, |Y| \, |Z|$.
    \item $K \leq \exp(d^7 \poly(\epsilon^{-1}))$.
\end{enumerate}
The algorithm is deterministic and runs in time $n^2 \cdot \exp(d^7 \poly(\epsilon^{-1}))$ (where $n = |X| + |Y| + |Z|$).
\end{restatable}

To illustrate how these four properties become useful, we complete the description of the Triangle Detection algorithm. We first precompute the decomposition as in the theorem. Additionally, define~$C_k = C[X_k, Z_k]$. Property~(1) of the theorem states that $A B = \sum_k A_k B_k$ (here, by slight abuse of notation, in the sum we interpret each term $A_k B_k$ as the $X \times Z$-matrix by extending $A_k B_k$ with zeros). Therefore, the set of triangles in the original graph is exactly the disjoint union of the triangles in the tripartite subgraphs $(X_k, Y_k, Z_k, A_k, B_k, C_k)$. It thus remains to detect a triangle in any of these subgraphs.

For each such subgraph, we are again in a win-win situation: If at least one of the edge parts is sparse (i.e., $\Ex[A_k] \leq 2^{-d}$ or $\Ex[B_k] \leq 2^{-d}$ or $\Ex[C_k] \leq 2^{-\epsilon d / 2}$), then we can solve the subinstance efficiently in time $|X_k| \, |Y_k| \, |Z_k| \,/\, 2^{\Omega(d)}$. Otherwise, Property~(2) of the theorem implies that we are in the dream scenario that $A_k$ and $B_k^T$ are $(\epsilon, 2, d)$-regular and $\epsilon$-min degree. Following the same argument as before, building on the structural \cref{thm:regular-product}, it follows that $A_k B_k$ and $C_k$ share a common nonzero entry, which entails the existence of a triangle. In summary, the algorithm solves each sparse subinstance in time $|X_k| \, |Y_k| \, |Z_k| \,/\, 2^{\Omega(d)}$ and stops as soon as it encounters a dense subinstance.

The remaining Properties~(3) and~(4) are necessary to bound the running time of this algorithm. On the one hand, by Property~(3) solving all sparse instances takes total time
\begin{equation*}
    \sum_{k=1}^K \frac{|X_k| \, |Y_k| \, |Z_k|}{2^{\Omega(d)}} \leq |X| \, |Y| \, |Z| \cdot \frac{\poly(d)}{2^{\Omega(d)}} \leq \frac{n^3}{2^{\Omega(d)}}.
\end{equation*}
On the other hand, precomputing the decomposition, and testing for each subinstance, whether it is dense or sparse, takes time~\smash{$n^2 \cdot \exp(d^7 \poly(\epsilon^{-1})) = n^2 \cdot 2^{\Order(d^7)}$}. By choosing $d = \Theta(\sqrt[7]{\log n})$ sufficiently small such that the precomputation time becomes $\Order(n^{2.1})$, say, the total running time becomes \smash{$n^3 / 2^{\Omega(\sqrt[7]{\log n})}$} as claimed.

The remainder of this overview is devoted to a proof overview of \cref{thm:2-path-decomposition}.

\subsection{Enforcing Regularity and Min-Degree}
Towards proving the decomposition theorem, our first milestone is to develop tools to enforce the (a) regularity and (b) min-degree conditions.

Both tools follow a common theme: To achieve some property we either certify that (a large part of) the given graph already satisfies the property, or that we can alternatively find a large induced subgraph which is substantially \emph{denser} than average (\emph{density increment}). In the former case we have been successful, and in the latter case we will simply \emph{recurse} on the selected denser piece. Since the density increases with each recursive call, we control the recursion depth and the loss we thereby incur. More details follow in \cref{sec:overview:sec:A-decomposition,sec:overview:sec:AB-decomposition}.

\paragraph{Enforcing Min-Degree}
Let us start with the conceptually easier min-degree property. Here, specifically, we would like to ensure the $\epsilon$-min-degree condition, i.e. that all nodes $x \in X$ satisfy $\deg_A(x) \geq (1 - \epsilon) \Ex[A]$. In fact, it is enough for us if we can find a subgraph of, say, half the total size that satisfies that it is $\epsilon$-min-degree. An easy algorithm is to repeatedly remove low-degree nodes~$x$ until the remaining graph becomes $\epsilon$-min-degree. If this algorithm terminates before removing half the nodes, then we have succeeded in finding a large $\epsilon$-min-degree subgraph. If instead the algorithm removes half the nodes in $X$ and the graph $A'$ is still not $\epsilon$-min-degree, then we claim that the remaining graph has density $\Ex[A'] \geq (1 + \frac{\epsilon}{2}) \Ex[A]$---i.e., we have found a density increment.
% Indeed, supposing that $\Ex[A'] < (1 + \frac{\epsilon}{2}) \Ex[A]$ leads to the contradiction that $A$ has density less than~\smash{$\frac12 \cdot (1 - \epsilon)(1 + \frac{\epsilon}{2}) \Ex[A] + \frac12 \cdot (1 + \frac{\epsilon}{2}) \Ex[A] < \Ex[A]$}. \shachar{I think this last sentence is too technical for here, and I suggest we remove it}\nick{Yeah, agreed.}
For more details, see \cref{lem:min-degree}.

\paragraph{Enforcing Regularity}
The more challenging task is to ensure that a graph is regular (or that we can alternatively find a density increment). Following the terminology from~\cite{KelleyLM23} (which in turn originates from~\cite{KelleyM23}), we refer to this step as ``sifting''. Specifically, we rely on the following theorem that we will later apply with $k = 2$ and $\ell = d$:

\begin{restatable}[Sifting]{theorem}{thmsifting} \label{thm:sifting}
Let $A \in \set{0, 1}^{X \times Y}$, let $\epsilon > 0$ and $k, \ell \geq 1$. There is an algorithm $\Sift(X, Y, A, \epsilon, k, \ell)$ that returns either
\begin{enumerate}
    \item ``regular'', in which case $A$ is $(\epsilon, k, \ell)$-regular, or
    \item sets $X' \subseteq X, Y' \subseteq Y$ with $|X'| \, |Y'| \geq \frac{\epsilon}{16} \cdot \Ex[A]^{k \ell} \cdot |X| \, |Y|$ and $\Ex[A[X', Y']] \geq (1 + \frac{\epsilon}{2}) \Ex[A]$.
\end{enumerate}
The algorithm is deterministic and runs in time~\smash{$n^{2} \cdot (\epsilon \Ex[A] / k)^{-\Order(k \ell(k + \ell))}$} (where $n = |X| + |Y|$).
\end{restatable}

The existential claim of \cref{thm:sifting} was already established by Kelley, Lovett and Meka~\cite{KelleyLM23} (up to insignificant changes in the parameters) by a simple ``one-shot'' proof. While it is possible to turn their ideas into a randomized sifting algorithm, we follow a different proof that can ultimately be turned into a \emph{deterministic} algorithm. The rough idea is to prove that whenever $A$ is not $(\epsilon, k, \ell)$-regular, then either many nodes $x \in X$ have exceptionally high degree $\deg_A(x) \geq (1 + \frac{\epsilon}{2}) \Ex[A]$ (in which case we can return the set $X'$ of such high-degree nodes and $Y' = Y$), or we can find a large induced subgraph of $A$ that is $(\epsilon, k - 1, \ell)$-irregular (see \cref{lem:sifting-recursion}). In the latter case, we recurse on that subgraph, so after at most $k$ recursive calls we find a density increment.

In order to detect this exceptionally irregular subgraph, it is necessary to obtain an accurate estimate of its grid norm. To this end, we prove that any grid norm $\norm{A}_{U(k, \ell)}$ can be approximated up to some additive error $\alpha > 0$ in time \smash{$n^2 \cdot \alpha^{-\Order(k \ell(k + \ell))}$} by a deterministic algorithm (\cref{lem:grid-norm-sampling}). Here we crucially build on the technology of \emph{oblivious samplers}. We defer further details to \cref{sec:sifting}.

% \todo{Recall that we only need the $k = 2$ case?}\shachar{we already mentioned it above, we repeat it here?}\nick{Just forgot to remove the todo}

\subsection{Decomposing \texorpdfstring{\boldmath$A$}{A}} \label{sec:overview:sec:A-decomposition}
As a warm-up and building block towards \cref{thm:2-path-decomposition}, let us first focus on decomposing a single bipartite graph $A \in \set{0, 1}^{X \times Y}$. Specifically, we establish the following decomposition with four analogous properties to \cref{thm:2-path-decomposition}.

\begin{restatable}[$A$-Decomposition]{theorem}{thmedgedecomposition} \label{thm:edge-decomposition}
Let $A \in \set{0, 1}^{X \times Y}$, let~\makebox{$\epsilon \in (0, 1)$} and $d \geq 1$. There is an algorithm $\ADecomposition(X, Y, A, \epsilon, d)$ computing tuples $\set{(X_\ell, Y_\ell, A_\ell)}_{\ell=1}^L$ with~\makebox{$X_\ell \subseteq X$},~\makebox{$Y_\ell \subseteq Y$}, and~$A_\ell \in \set{0, 1}^{X_\ell \times Y_\ell}$ such that:
\begin{enumerate}
    \item $A = \sum_{\ell=1}^L A_\ell$.
    \item For all $\ell \in [L]$:
    \begin{enumerate}[label=(\roman*)]
        \item $\Ex[A_\ell] \leq 2^{-d}$, or
        \item $A_\ell$ is $(\epsilon, 2, d)$-regular and $\epsilon$-min-degree.
    \end{enumerate}
    \item $\sum_{\ell=1}^L |X_\ell| \, |Y_\ell| \leq (d + 2) \cdot |X| \, |Y|$.
    \item $L \leq \exp(d^3 \poly(\epsilon^{-1}))$ and $\min_{\ell=1}^L |X_\ell| \, |Y_\ell| \geq \exp(-d^3 \poly(\epsilon^{-1})) \cdot |X| \, |Y|$.
\end{enumerate}
The algorithm is deterministic and runs in time $n^2 \cdot \exp(d^3 \poly(\epsilon^{-1}))$ (where $n = |X| + |Y|$).
\end{restatable}

% \todo{Hype the theorem a little bit.}

\cref{thm:edge-decomposition} in itself is already an interesting regularity decomposition, which we believe will likely find further applications in the future. The proof of the theorem is along the following lines. First of all, if $\Ex[A] \leq 2^{-d}$, then we simply return the trivial decomposition $\set{(X, Y, A)}$. So assume from now on that $\Ex[A] \geq 2^{-d}$.

Consider the following subtask (see \cref{lem:good-rect}): The goal is to find $X^* \subseteq X$ and $Y^* \subseteq Y$ such that the induced subgraph $A[X^*, Y^*]$ is $(\epsilon, 2, d)$-regular and $\epsilon$-min-degree; we call $X^* \times Y^*$ a \emph{good rectangle}. We can find a good rectangle using the density increment technique. First, make half of~$X$ satisfy the min-degree condition. Then, apply \cref{thm:sifting} to certify that this remaining graph is $(\epsilon, 2, d)$-regular. If both steps succeed we have successfully identified a good rectangle (namely, the entire remaining graph). Otherwise, if either step fails and instead returns a large subgraph with density at least~\makebox{$(1 + \frac{\epsilon}{2}) \Ex[A]$}, we simply \emph{recurse} on the denser subgraph to find a good rectangle. With each recursive call the density strictly increases, and thus this process eventually terminates. In fact, the recursion depth is bounded by $\Order(d / \epsilon)$ given that the initial density is $\Ex[A] \geq 2^{-d}$. Since with each recursive call we reduce the number of vertices to a~$(\epsilon \Ex[A])^{-\Order(d)} = \exp(-d^2 \poly(\epsilon^{-1}))$-fraction (by \cref{thm:sifting}), the returned good rectangle covers at least a $\exp(-d^3 \poly(\epsilon^{-1}))$-fraction of the original graph.

Coming back to \cref{thm:edge-decomposition}, we can compute the decomposition using density \emph{decrements}. Namely, we repeatedly find good rectangles $X^* \times Y^*$ as in the previous paragraph, take the subgraph~$(X^*, Y^*, A[X^*, Y^*])$ as one part in the decomposition, and then remove all edges in the rectangle $X^* \times Y^*$ from $A$. Eventually $A$ becomes $2^{-d}$-sparse and at this point we return the remaining trivial decomposition $\set{(X, Y, A)}$. In each step we remove $\exp(-d^3 \poly(\epsilon^{-1})) \cdot |X| \cdot |Y|$ edges from $A$, and therefore this process leads to at most $L \leq \exp(d^3 \poly(\epsilon^{-1}))$ many parts.

So far we have neglected Property~3, but it turns out that a closer inspection of this process indeed yields that $\sum_{\ell=1}^L |X_\ell| \, |Y_\ell| \leq \Order(d) \cdot |X| \, |Y|$. Proving this statement builds on the critical insight that, for any good rectangle that we remove, we always have $\Ex[A[X^*, Y^*]] \geq \Ex[A]$. Specifically, let $\overline A_\ell$ denote the remaining matrix $A$ before the $\ell$-th step of the algorithm, and let $L_1$ be the smallest index such that $\Ex[\overline A_{L_1}] \leq \frac{1}{2} \Ex[A]$. Then we can express
\begin{equation*}
    \Ex[A] = \sum_{\ell=1}^{L_1-1} \Ex[A_\ell] \cdot \frac{|X_\ell| \, |Y_\ell|}{|X| \, |Y|} + \Ex[\overline A_{L_1}] \geq \frac{\Ex[A]}{2} \cdot \sum_{\ell=1}^{L_1 - 1} \frac{|X_\ell| \, |Y_\ell|}{|X| \, |Y|},
\end{equation*}
using that $\Ex[A_\ell] \geq \Ex[\overline A_\ell] > \frac{1}{2} \Ex[A]$. It follows that $\sum_{\ell=1}^{L_1 - 1} |X_\ell| \, |Y_\ell| \leq 2 \cdot |X| \, |Y|$. We can apply the same argument to analyze the algorithm in \emph{phases}. That is, letting $L_i$ be the smallest step with~\smash{$\Ex[\overline A_{L_i}] \leq \frac{1}{2^i} \Ex[A]$}, we can similarly bound $\sum_{\ell=L_i}^{L_{i+1}-1} |X_\ell| \, |Y_\ell| \leq 2 \cdot |X| \, |Y|$ for each phase. After the $d$-th phase the process has reduced the density of the remaining graph to at most~\smash{$2^{-d}$}, and the process terminates. Therefore, all in all, we have $\sum_{\ell=1}^L |X_\ell| \, |Y_\ell| \leq 2d \cdot |X| \, |Y|$. (In the formal proof we obtain a slightly sharper bound.)

\subsection{Decomposing \texorpdfstring{\boldmath$AB$}{AB}} \label{sec:overview:sec:AB-decomposition}
We finally turn to the full decomposition from \cref{thm:2-path-decomposition}. The idea is to use the one-part decomposition developed in the previous subsection as a black-box to decompose $A$, and to decompose~$B$ via density increments/decrements. Unfortunately, the details of this step are much more intricate.

Let us first sketch an approach that will not work out as planned. In light of the previous subsection, the hope is that we can find a good rectangle $Y^* \times Z^*$ in $B$ along with a regularity decomposition $\set{(X_\ell, Y_\ell, A_\ell)}_{\ell=1}^L$ of $A[X, Y^*]$, such that additionally each subgraph $B[Y_\ell, Z^*]$ is $(\epsilon, 2, d)$-regular and $\epsilon$-min-degree (see \cref{lem:good-cube}). We loosely refer to $Y^*, Z^*, \set{(X_\ell, Y_\ell, A_\ell)}_{\ell=1}^L$ as a \emph{good cube}. If there was an algorithm to find good cubes, then we would easily obtain the desired decomposition: We repeatedly find a good cube $Y^*, Z^*, \set{(X_\ell, Y_\ell, A_\ell)}_{\ell=1}^L$, emit $\set{(X_\ell, Y_\ell, Z^*, A_\ell, B[Y_\ell, Z^*])}_{\ell=1}^L$ as parts in the decomposition and remove the edges in $Y^* \times Z^*$ from $B$.

However, we face serious problems trying to find a good cube. The natural idea is to use the density increment technique to find $Y^* \times Z^* \subseteq Y \times Z$ such that $B[Y^*, Z^*]^T$ is $(\epsilon, 2, d)$-regular and $\epsilon$-min-degree. Then we can apply \cref{thm:edge-decomposition} to decompose the matrix $A[X, Y^*]$ into pieces~$\set{(X_\ell, Y_\ell, A_\ell)}_{\ell=1}^L$. However, the subgraphs $B[Y_\ell, Z^*]$ are not necessarily $(\epsilon, 2, d)$-regular and $\epsilon$-min-degree. The first issue is fixable: Using the sifting algorithm we can test whether all subgraphs $B[Y_\ell, Z^*]$ are $(\epsilon, 2, d)$-regular---if any such subgraphs fails this test, then \cref{thm:sifting} instead returns a denser subgraph of $B[Y_\ell, Z^*]$. We thus recurse on that subgraph to find a good cube. The second issue, that $B[Y_\ell, Z^*]$ is not $\epsilon$-min-degree, is more serious. In contrast to the regularity condition we cannot enforce the min-degree condition for the whole graph $B[Y_\ell, Z^*]$, but only for a subgraph~\makebox{$B[Y_\ell, Z_\ell]$}, where $Z_\ell \subseteq Z^*$ is large.

Our solution to this issue is somewhat reminiscent to the divide-and-conquer approaches for Triangle Detection. As outlined before, we can compute a good cube $Y^*, Z^*, \set{(X_\ell, Y_\ell, Z_\ell, A_\ell)}$ such that each subgraph $B[Y_\ell, Z_\ell]$ is $(\epsilon, 2, d)$-regular and $\epsilon$-min-degree. By tweaking the parameters of the min-degree lemma, we can further achieve that $|Z_\ell| \geq (1 - \gamma) |Z^*|$ for some parameter $\gamma$ to be determined soon. As before, for each good cube we emit $\set{(X_\ell, Y_\ell, Z_\ell, A_\ell, B[Y_\ell, Z_\ell])}$ as one part of the decomposition, and then remove the edges in $Y^* \times Z^*$ from $B$ and repeat. However, we additionally recurse on all subinstances on the vertex parts $(X_\ell, Y_\ell, Z^* \setminus Z_\ell)$ to cover the edges missed in the previous parts. To control the cost caused by this additional layer of recursion we need to guarantee that
\begin{equation*}
    \sum_{\ell=1}^L |X_\ell| \, |Y_\ell| \, |Z^* \setminus Z_\ell| \leq \tfrac12 \cdot |X| \, |Y| \, |Z|,
\end{equation*}
say. And indeed, using that $\sum_{\ell=1}^L |X_\ell| \, |Y_\ell| \leq \poly(d) \cdot |X| \, |Y|$ and by setting $\gamma = \frac{1}{\poly(d)}$ small enough this can be achieved. The overhead nevertheless leads to a significant blow-up in the dependence on $d$, from $\exp(d^3)$ to $\exp(d^7)$. See \cref{sec:decompositions} for the details.

\subsection{Further Improvements?}
%\paragraph{Improving the Parameters?}
While we have successfully achieved quasi-polynomial savings combinatorially for BMM (and many other problems), \cref{conj:bmm} still stands, and the question remains whether we can do better. E.g., can we achieve savings of the form~\smash{$2^{\Theta(\sqrt{\log n})}$} rather than \smash{$2^{\Theta(\sqrt[7]{\log n})}$}, or possibly even truly polynomial savings? 

Over random matrices in which each entry is $1$ with probability $p$, we can solve BMM in~$\widetilde{O}(n^{2.5})$ time,\footnote{For $p \gg 1/\sqrt{n}$ the answer is the all-ones matrix with good probability, and otherwise the matrices are sparse.} which means that the general framework of worst-case to random-case reductions via regularity decompositions could go much further.
However, it is not clear whether the specific notion of grid regularity can go beyond \smash{$2^{\Theta(\sqrt{\log n})}$} savings, because the known lower bounds for Triangle Removal (à la Behrend's construction~\cite{Behrend46,Zhao23}) seem to apply as well.
We have focused on presenting our new technique in an easy and modular fashion rather than on optimizing the constant in the quasi-polynomial savings. It remains an interesting open question to fine-tune the parameters.
%, and it is likely to give the \smash{$n^3/2^{\Theta(\sqrt{\log n})}$} bound for BMM. % \amir{it would be nice to explain that $2^{\sqrt{\log{n}}}$ is a limit for this notion of regularity.}
% !TEX root = ../paper.tex
\section{Sifting} \label{sec:sifting}
In this section we describe the ``sifting'' algorithm that, given a graph $A$, either determines that~$A$ is regular or finds a subgraph of $A$ that is denser than average. Kelley, Lovett and Meka~\cite{KelleyLM23} have proved this statement via a non-algorithmic proof that can rather easily be turned into a randomized algorithm. Our approach here differs from that original version, as our more ambitious goal is to obtain a \emph{deterministic} sifting algorithm. We start with a simple inverse of Markov's inequality:

\begin{lemma}[Inverse of Markov's Inequality] \label{lem:inverse-markov}
Let $Z$ be a random variable that takes values in~$[0, 1]$. Then, for any $\alpha \in [0, 1]$,
\begin{equation*}
    \Pr[Z \geq \alpha] \geq \Ex[Z] - \alpha.
\end{equation*}
\end{lemma}
\begin{proof}
Observe that $\One(Z \geq \alpha) \geq Z - \alpha$. Thus, $\Pr[Z \geq \alpha] = \Ex[\One(Z \geq \alpha)] \geq \Ex[Z] - \alpha$.
\end{proof}

The high-level idea behind the sifting algorithm is that we can either (1) find a denser subgraph by simply taking the high-degree nodes, or (2) recurse on a smaller subgraph with parameter~\makebox{$k - 1$}. This idea is recorded in the next lemma. Here and for the remainder of this section we write~\makebox{$Y_x = \set{y \in Y : A(x, y) = 1}$} and $A_x = A[X, Y_x]$.

\begin{lemma}[Recursive Sifting] \label{lem:sifting-recursion}
Let $A \in \set{0, 1}^{X \times Y}$, let $\delta, \epsilon > 0$ and $k, \ell \geq 1$ and assume that $\norm{A}_{U(k, \ell)} \geq (1 + \epsilon) \delta$. Then one of the following two cases applies:
\begin{enumerate}
\item $\abs{\set{x \in X : \deg_A(x) \geq \delta}} \geq \tfrac{\epsilon}{2} \cdot \delta^{k \ell} \cdot |X|$,
\item or $k > 1$ and there is some $x \in X$ such that:
\begin{itemize}
    \item $\deg_A(x) \geq \delta^k$, and
    \item $\norm{A_x}_{U(k-1, \ell)} \geq (1 + \epsilon) \delta$.
\end{itemize}
\end{enumerate}
\end{lemma}
\begin{proof}
First consider the case $k = 1$. We prove that case 1 applies by sampling $x \in X$ uniformly at random, and showing that with probability at least $\epsilon \cdot \delta^\ell$ this choice satisfies $\deg_A(x) \geq \delta$. Indeed, using the inverse Markov inequality:
\begin{align*}
    \Pr_{x \in X}\brackets*{\deg_A(x)^\ell \geq \delta^\ell} & \geq \Ex_{x \in X} \deg_A(x)^\ell - \delta^\ell = \Ex_{x \in X} \parens*{\Ex_{y \in Y} A(x, y)}^\ell - \delta^\ell \\&= \norm{A}_{U(1, \ell)}^{\ell} - \delta^\ell \geq (1 + \epsilon)^\ell \delta^\ell - \delta^\ell \geq \epsilon \cdot \delta^\ell.
\end{align*}

Next, let $k > 1$ and suppose that case 1 does not hold. We prove that selecting a uniformly random element~\makebox{$x \in X$} satisfies case~2 with positive probability. In fact, we prove that a uniformly random element~\makebox{$x \in X$} satisfies the following two stronger properties with positive probability:
\begin{enumerate}[label=(\roman*)]
    \item $\deg_A(x) \leq \delta$,
    \item $\deg_A(x) \cdot \norm{A_x}_{U(k-1, \ell)}^{k-1} \geq (1 + \epsilon)^{k-1} \delta^k$.
\end{enumerate}
Clearly, (i) fails with probability at most $\frac{\epsilon}{2} \cdot \delta^{k \ell}$ (by the assumption that case 1 of the lemma statement does not hold). Considering property~(ii), we first bound the following expectation:
\begin{align*}
    &\Ex_{x \in X} \deg_A(x)^\ell \cdot \norm{A_x}_{U(k-1, \ell)}^{(k-1)\ell} \\
    &\qquad= \Ex_{x \in X} \Ex_{x_2, \dots, x_k \in X} \parens*{\deg_A(x) \cdot \Ex_{y \in Y_x} \prod_{i=2}^k A(x_i, y)}^\ell
\intertext{Recall that $Y_x$ is the set of neighbors of $x$. Hence, for any function $f$ we can rewrite the expectation~\makebox{$\deg_A(x) \cdot \Ex_{y \in Y_x} f(y)$} as $\Ex_{y \in Y} f(y) \cdot A(x, y)$, and thus}
    &\qquad= \Ex_{x_1, \dots, x_k \in X} \parens*{\Ex_{y \in Y} \prod_{i \in [k]} A(x_i, y)}^\ell \\[1ex]
    &\qquad= \Ex_{\substack{x_1, \dots, x_k \in X\\y_1, \dots, y_\ell \in Y}} \prod_{\substack{i \in [k]\\j \in [\ell]}} A(x_i, y_j) \\[1ex]
    &\qquad= \norm{A}_{U(k, \ell)}^{k \ell}.
\end{align*}
Therefore, by the inverse Markov inequality, for a uniformly random $x \in X$ property~(ii) holds with probability at least
\begin{align*}
    &\Pr_{x \in X} \brackets*{\deg_A(x) \cdot \norm{A_x}_{U(k-1, \ell)}^{k-1} \geq (1 + \epsilon)^{k-1} \delta^k} \\
    &\qquad= \Pr_{x \in X} \brackets*{\deg_A(x)^\ell \cdot \norm{A_x}_{U(k-1, \ell)}^{(k-1)\ell} \geq (1 + \epsilon)^{(k-1)\ell} \delta^{k \ell}} \\
    &\qquad\geq \Ex_{x \in X} \deg_A(x)^\ell \cdot \norm{A_x}^{(k-1)\ell}_{U(k-1, \ell)} - (1 + \epsilon)^{(k-1)\ell} \delta^{k \ell} \vphantom{\Ex_X\Big[}\\
    &\qquad\geq \norm{A}_{U(k, \ell)}^{k \ell} - (1 + \epsilon)^{(k-1)\ell} \delta^{k \ell}. \vphantom{\Ex_X\Big[}\\
    &\qquad\geq (1 + \epsilon)^{k \ell} \delta^{k \ell} - (1 + \epsilon)^{(k-1)\ell} \delta^{k \ell} \vphantom{\Ex_X\Big[}\\
    &\qquad\geq \epsilon \cdot \delta^{k \ell}. \vphantom{\Ex_X\Big[}
\end{align*}
By a union bound, both properties~(i) and~(ii) hold simultaneously with positive probability at least~\smash{$\frac{\epsilon}{2} \cdot \Ex[A]^{k \ell}$}.

Finally, consider an element $x$ satisfying (i) and (ii); we show that $x$ also satisfies the two conditions from the lemma statement. Since \smash{$\delta \cdot \norm{A_x}_{U(k-1, \ell)}^{k-1} \geq \deg_A(x) \cdot \norm{A_x}_{U(k-1, \ell)}^{k-1} \geq (1 + \epsilon)^{k-1} \delta^k$}, it follows that indeed $\norm{A_x}_{U(k-1, \ell)} \geq (1 + \epsilon) \delta$. Moreover, using the trivial bound $\norm{A_x}_{U(k-1, \ell)} \leq 1$, we conclude that~\smash{$\deg_A(x) \geq \delta^k$}.
\end{proof}

For the sifting algorithm we also need the following lemma about approximating $\norm{\,\cdot\,}_{U(k, \ell)}$. We postpone the deterministic proof of \cref{lem:approx-grid-norm} to \cref{sec:sifting:sec:approx-grid-norm}, and encourage the reader to instead think of \cref{lem:approx-grid-norm} as the straightforward randomized algorithm (that subsamples $X^k \times Y^\ell$ to approximately count the number of $(k, \ell)$-bicliques).

\begin{lemma}[Deterministic Regularity Approximation] \label{lem:approx-grid-norm}
Let $A \in \set{0, 1}^{X \times Y}$, let $\alpha > 0$ and~\makebox{$k, \ell \geq 1$}. There is a deterministic algorithm that computes, for all~\makebox{$x \in X$}, an approximation $v_x$ satisfying that~\smash{$v_x = \norm{A_x}_{U(k, \ell)} \pm (\alpha / \deg_A(x)^{\frac{1}{k}})$}, and runs in time~\smash{$n^2 \cdot \alpha^{-\Order(k \ell (k + \ell))}$} (where $n = |X| + |Y|$).
\end{lemma}

We are ready to prove \cref{thm:sifting}. For convenience, we restate the statement here.

\begin{algorithm}[t]
\caption{Implements the algorithm from \cref{thm:sifting}.} \label{alg:sifting}
\begin{algorithmic}[1]
\Procedure{\Siftp}{$X, Y, A, \delta, \epsilon, k, \ell$}
    \State Let $X' \gets \set{x \in X : \deg_A(x) \geq \delta}$ \label{alg:sifting:line:high-deg}
    \If{$|X'| \geq \frac{\epsilon}{2} \cdot \delta^{k \ell} \cdot |X|$} \label{alg:sifting:line:high-deg-cond}
        \State\Return $X', Y$ \label{alg:sifting:line:high-deg-return}
    \EndIf
    \If{$k = 1$} \label{alg:sifting:line:k=1}
        \State\Return ``regular'' \label{alg:sifting:line:k=1-return}
    \Else \label{alg:sifting:line:k>1}
        \State Compute approximations $v_x$ of $\norm{A_x}_{U(k-1, \ell)}$ by \cref{lem:approx-grid-norm} with parameter \smash{$\alpha = \frac{\epsilon\delta^2}{2k^2}$} \label{alg:sifting:line:approx-grid-norm}
        \State Select $x \in X$ with \smash{$\deg_A(x) \geq \delta^k$} maximizing $v_x$\label{alg:sifting:line:select-x}
        \State\Return $\Siftp(X, Y_x, A_x, \delta, \epsilon \cdot (1 - \frac{1}{k^2}), k - 1, \ell)$ \label{alg:sifting:line:recur}
    \EndIf
\EndProcedure

\bigskip
\Procedure{\Sift}{$X, Y, A, \epsilon, k, \ell$}
    \State\Return $\Siftp(X, Y, A, (1 + \tfrac{\epsilon}{2}) \Ex[A], \tfrac{\epsilon}{4}, k, \ell)$
\EndProcedure
\end{algorithmic}
\end{algorithm}

\thmsifting*
\begin{proof}
Throughout we assume that $k \leq \ell$, as otherwise we can simply exchange $k$ and $\ell$ and work on the transposed graph $A^T$. Moreover, for the sake of a cleaner presentation we design an algorithm $\Siftp(X, Y, A, \delta, \epsilon, k, \ell)$ that receives an additional input parameter $\delta > 0$ with the modified task to return either
\begin{enumerate}
    \item sets $X' \subseteq X, Y' \subseteq Y$ with $|X'| \, |Y'| \geq \frac{\epsilon}{4} \cdot \delta^{k \ell} \cdot |X| \, |Y|$ and $\Ex[A[X', Y']] \geq \delta$, or
    \item ``regular'', in which case $\norm{A}_{U(k, \ell)} \leq (1 + \epsilon) \delta$.
\end{enumerate}
From this alternative algorithm we can easily obtain the desired algorithm $\Sift(X, Y, A, \epsilon, k, \ell)$: Simply call and return $\Siftp(X, Y, A, (1 + \frac{\epsilon}{2}) \Ex[A], \frac{\epsilon}{4}, k, \ell)$.

We design $\Siftp(X, Y, A, \delta, \epsilon, k, \ell)$ as a simple recursive algorithm; see \cref{alg:sifting} for the pseudocode. In a first step (\crefrange{alg:sifting:line:high-deg}{alg:sifting:line:high-deg-return}) we construct the set $X' \gets \set{x \in X : \deg_A(x) \geq \delta}$ of high-degree nodes. If this set turns out to be sufficiently large, $|X'| \geq \frac{\epsilon}{2} \cdot \delta^{k \ell} \cdot |X|$, we can successfully return~$X'$ and~\makebox{$Y' \gets Y$}. Otherwise, we distinguish two cases: If $k = 1$, then we simply report ``regular'' (\crefrange{alg:sifting:line:k=1}{alg:sifting:line:k=1-return}). If instead $k > 1$, then our principle strategy is to identify a node~\makebox{$x \in X$} such that subgraph~$A_x$ is \emph{as irregular as possible,} and to recurse on that subgraph (\crefrange{alg:sifting:line:approx-grid-norm}{alg:sifting:line:recur}). Specifically, using \cref{lem:approx-grid-norm} we compute approximations~$v_x$ of~$\norm{A_x}_{U(k-1, \ell)}$, for all~\makebox{$x \in X$}, with parameter~\smash{$\alpha = \frac{\epsilon\delta^2}{2k^2}$}. We then select the element $x \in X$ with~\smash{$\deg_A(x) \geq \delta^k$} that maximizes $v_x$. Finally, we recurse on \smash{$\Siftp(X, Y_x, A_x, \delta, \epsilon \cdot (1 - \frac{1}{k^2}), k - 1, \ell)$}. (Here we tweak the parameter $\epsilon$ to account for the loss in the approximation.)

\paragraph{Correctness}
First observe that if the algorithm returns sets $X', Y'$, then indeed~\makebox{$\Ex[A[X', Y']] \geq \delta$}. Next we verify by induction that~\smash{$|X'| \, |Y'| \geq \frac{\epsilon}{2} \cdot (\prod_{i=2}^k (1 - \frac{1}{i^2})) \cdot \delta^{k \ell} \cdot |X| \, |Y|$} (which yields the desired bound using that~\smash{$\prod_{i=2}^k (1 - \frac{1}{i^2}) = \frac{k+1}{2k} \geq \frac12$}). Indeed, in the base case we return the sets $X'$ and $Y' = Y$ with size at least $|X'| \, |Y'| \geq \frac{1}{2} \cdot \delta^{k \ell} \cdot |X| \, |Y|$. Otherwise, we recursively obtain sets of size \smash{$|X'| \, |Y'| \geq \frac{\epsilon}{2} \cdot (1 - \frac{1}{k^2}) \cdot (\prod_{i=2}^{k-1} (1 - \frac{1}{i^2})) \cdot \delta^{(k-1) \ell} \cdot |X| \, |Y_x|$}. Recall that $|Y_x| \geq \delta^k \cdot |Y|$, which completes the claim.

It remains to prove that whenever the given graph is irregular, i.e.~\makebox{$\norm{A}_{U(k, \ell)} \geq (1 + \epsilon) \delta$}, then our algorithm does not return ``regular''. On the one hand, if $k = 1$ then \cref{lem:sifting-recursion} implies that the set of high-degree nodes $X' = \set{x \in X : \deg_A(x) \geq \delta}$ has size at least $\frac{\epsilon}{2} \cdot \delta^{\ell} \cdot |X|$, and therefore the algorithm terminates in \cref{alg:sifting:line:high-deg-return} by returning $X', Y$. On the other hand, consider the case $k > 1$. Then either the algorithm terminates in \cref{alg:sifting:line:high-deg-return}, or \cref{lem:sifting-recursion} implies that there is some node~$x^*$ such that $\deg_A(x^*) \geq \delta^k$ and $\norm{A_{x^*}}_{U(k-1, \ell)} \geq (1 + \epsilon) \delta$. Since we select $x$ to be the element satisfying~\smash{$\deg_A(x) \geq \delta^k$} maximizing the approximation of \smash{$\norm{A_x}_{U(k-1, \ell)}$} with additive error~\smash{$\alpha / \deg_A(x)^{\frac{1}{k}} \leq \frac{\epsilon\delta}{2k^2}$}, we select an element with $\norm{A_x}_{U(k-1, \ell)} \geq (1 + \epsilon) \delta - \frac{\epsilon \delta}{k^2} = (1 + \epsilon \cdot (1 - \frac{1}{k^2})) \delta$. By induction the recursive call does not return ``regular''.

\paragraph{Running Time}
The recursion depth of the algorithm is at most $k$, hence it suffices to bound the running time of a single execution. The only costly step is the computation of the approximations of $\norm{A_x}_{U(k-1, \ell)}$ which takes time $n^2 \cdot \alpha^{-\Order(k \ell (k + \ell))} = n^2 \cdot (\epsilon\delta/k)^{-\Order(k \ell (k + \ell))}$ by \cref{lem:approx-grid-norm}; all other steps can be implemented in time $\Order(n^2)$.
\end{proof}

\subsection{Deterministic Regularity Approximation} \label{sec:sifting:sec:approx-grid-norm}
To obtain a deterministic algorithm, it remains to prove \cref{lem:approx-grid-norm}. As the key tool in our derandomization we rely on \emph{oblivious samplers} as developed in an extensive line of research~\cite{ChorG89,GoldreichW97,Zuckerman97,Gillman98,ReingoldVW00,GuruswamiUV09} (see also the survey~\cite{Goldreich11}). For our application the exact dependence on the accuracy parameters~$\epsilon, \delta$ does not matter much, and we thus rely on one of the early constructions:\footnote{Let us restate \cref{lem:oblivious-sampler} in the language of samplers. We arbitrarily identify $X$ with $\set{0, 1}^n$ where~\makebox{$n = \ceil{\log |X|}$}. \cref{lem:oblivious-sampler} states that there is a randomized algorithm (an \emph{oblivious sampler}) returning a sample set~\makebox{$S \subseteq \set{0, 1}^n$} of size $|S| \leq \poly(\epsilon^{-1}, \delta^{-1})$ such that, for any function $f : \set{0, 1}^n \to [0, 1]$, with probability at least $1 - \delta$ it holds that~\smash{$\Ex_{x \in X} f(x) = \Ex_{x \in S} f(x) \pm \epsilon$}. Moreover, the algorithm runs in polynomial time $\poly(n, \epsilon^{-1}, \delta^{-1})$ and has randomness complexity $n + \Order(\log \epsilon^{-1}) + \Order(\log \delta^{-1})$ (i.e., it tosses at most that many unbiased coins). In our formulation, $\mathcal S$ denotes the set of all sample sets $S$ obtainable from the algorithm for some sequence of coin tosses.}

\begin{lemma}[Oblivious Sampling~\cite{GoldreichW97}] \label{lem:oblivious-sampler}
Let $X$ be a set and let $\delta, \epsilon > 0$. There is a deterministic algorithm computing, in time $|X| \cdot \poly(\epsilon^{-1}, \delta^{-1}, \log |X|)$, a family $\mathcal S$ of subsets $S \subseteq X$ such that
\begin{enumerate}
    \item $|\mathcal S| \leq |X| \cdot \poly(\epsilon^{-1}, \delta^{-1})$,
    \item $|S| \leq \poly(\epsilon^{-1}, \delta^{-1})$ for all $S \in \mathcal S$,
    \item For every function $f : X \to [0, 1]$,
    \begin{equation*}
        \Pr_{S \in \mathcal S} \brackets*{\Ex_{x \in X} f(x) = \Ex_{x \in S} f(x) \pm \epsilon} \geq 1 - \delta.
    \end{equation*}
\end{enumerate}
We call $\mathcal S$ an \emph{$(\epsilon, \delta)$-oblivious sampler of $X$.}
\end{lemma}

\begin{lemma}[Simple Bounds for Additive Approximations] \label{lem:bounds-add-approx}
Let $a, b \in [0, 1]$ and $\epsilon > 0, k \ge 1$. Then:
\begin{itemize}
    \item If $a = b \pm \epsilon$, then $a^k = b^k \pm 2\epsilon k$.
    \item If $a^k = b^k \pm \epsilon^k$, then $a = b \pm \epsilon$.
\end{itemize}
\end{lemma}
\begin{proof}
The first claim is trivial if \smash{$\epsilon \geq \frac{1}{2k}$}, so suppose otherwise. Assuming that $a \leq b + \epsilon$, it follows that $a^k - b^k \leq (b + \epsilon)^k - b^k = \sum_{i=1}^k \binom{k}{i} \cdot \epsilon^i \cdot b^{k-i} \leq \sum_{i=1}^k k^i \cdot \epsilon^i \leq \epsilon k \cdot \sum_{i=0}^\infty k^i \cdot \epsilon^i \leq 2\epsilon k$. The second claim is immediate: If $a > b + \epsilon$, then $a^k > (b + \epsilon)^k > b^k + \epsilon^k$ (and similarly if~\makebox{$b > a + \epsilon$}).
\end{proof}

\begin{lemma}[Regularity Approximation via Oblivious Sampling] \label{lem:grid-norm-sampling}
Let $A \in \set{0, 1}^{X \times Y}$, let $\delta, \epsilon > 0$ and~\makebox{$k, \ell \geq 1$}, and let $\mathcal S, \mathcal T$ be $(\epsilon, \delta)$-oblivious samplers of $X$ and $Y$, respectively. Then:
\begin{equation*}
    \norm{A}_{U(k, \ell)}^{k \ell} = \Ex_{\substack{S \in \mathcal S\\T \in \mathcal T}} \norm{A[S, T]}_{U(k, \ell)}^{k \ell} \pm (2\epsilon k + 2\epsilon \ell + 2\delta).
\end{equation*}
\end{lemma}
\begin{proof}
Fix $y_1, \dots, y_\ell \in Y$ and consider the function $f(x) = \prod_{j \in [\ell]} A(x, y_j)$. We sample a uniformly random set $S \in \mathcal S$. Since $\mathcal S$ is an $(\epsilon, \delta)$-oblivious sampler of $X$, with probability at least $1 - \delta$ we have that $\Ex_{x \in X} f(x) = \Ex_{x \in S} f(x) \pm \epsilon$ and thus $\parens{\Ex_{x \in X} f(x)}^k = \parens{\Ex_{x \in S} f(x)}^k \pm 2\epsilon k$ (by \cref{lem:bounds-add-approx}). Since moreover both expectations are $[0, 1]$-bounded, we conclude that
\begin{equation*}
    \parens*{\Ex_{x \in X} f(x)}^k = \Ex_{S \in \mathcal S} \parens*{\Ex_{x \in S} f(x)}^k \pm (2\epsilon k + \delta).
\end{equation*}
Now unfix $y_1, \dots, y_\ell \in Y$. From the previous consideration it follows that
\begin{align*}
    \norm{A}_{U(k, \ell)}^{k \ell}
    &= \Ex_{y_1, \dots, y_\ell \in Y} \parens*{\Ex_{x \in X} \prod_{j \in [\ell]} A(x, y_j)}^k \pm (2\epsilon k + \delta) \\
    &= \Ex_{S \in \mathcal S} \Ex_{y_1, \dots, y_\ell \in Y} \parens*{\Ex_{x \in S} \prod_{j \in [\ell]} A(x, y_j)}^k \pm (2\epsilon k + \delta) \\
    &= \vphantom{\parens*{\prod_{j \in [\ell]}}}\Ex_{S \in \mathcal S} \norm{A[S, Y]}_{U(k, \ell)}^{k \ell} \pm (2\epsilon k + \delta).
\end{align*}

We can now apply the same argument again to $A[S, Y]$, with the roles of $X$ and $Y$ interchanged, to obtain that
\begin{equation*}
    \norm{A[S, Y]}_{U(k, \ell)}^{k \ell} = \Ex_{T \in \mathcal T} \norm{A[S, T]}_{U(k, \ell)}^{k \ell} \pm (2\epsilon \ell + \delta),
\end{equation*}
and therefore
\begin{equation*}
    \norm{A}_{U(k, \ell)}^{k \ell} = \Ex_{\substack{S \in \mathcal S\\T \in \mathcal T}} \norm{A[S, T]}_{U(k, \ell)}^{k \ell} \pm (2\epsilon k + 2\epsilon \ell + 2\delta)
\end{equation*}
The claim follows.
\end{proof}

\begin{proof}[Proof of \cref{lem:approx-grid-norm}]
We precompute $(\epsilon, \delta)$-oblivious samplers $\mathcal S$ and $\mathcal T$ of $X$ and $Y$, respectively, for parameters $\epsilon, \delta > 0$ to be determined later. Then we enumerate each pair $S \in \mathcal S, T \in \mathcal T$ and each tuple $(x_1, \dots, x_k) \in S^k, (y_1, \dots, y_\ell) \in T^\ell$ to compute the values
\begin{equation*}
    u_{T, y_1, \dots, y_\ell} \gets \Ex_{S \in \mathcal S} \Ex_{x_1, \dots, x_k \in S} \prod_{\substack{i \in [k]\\j \in [\ell]}} A(x_i, y_j).
\end{equation*}
Next, we enumerate each $x \in X$ and compute%\shachar{it will simplify the proof if you define $u_x$ (say) below without dividing by $\deg_A(x)$, so you don't need to carry it all the way.}
\begin{equation*}
    u_x \gets \parens*{\Ex_{T \in \mathcal T} \Ex_{y_1, \dots, y_\ell \in T} u_{T, y_1, \dots, y_\ell} \cdot \prod_{j \in [\ell]} A(x, y_j)}^{\frac{1}{k \ell}},
\end{equation*}
by enumerating each $T \in \mathcal T$ and each tuple $(y_1, \dots, y_\ell) \in T^\ell$. Finally, we return \smash{$v_x \gets u_x / \deg_A(x)^{\frac{1}{k}}$} as the desired approximations.
% \shachar{I think you need to normalize by the degree of $x$, since $v_x$ is the norm of the matrix where you replace all other rows not in $A_x$ with zeros} \nick{I agree, done.}

For the correctness, let $A_x' \in \set{0, 1}^{X \times Y}$ be the matrix obtained from $A$ where all columns~\makebox{$y \not\in Y_x$} are zeroed out (in contrast to $A_x$ where we have deleted these columns). Then:
\begin{align*}
    u_x 
    &= \parens*{\Ex_{\substack{S \in \mathcal S\\T \in \mathcal T}} \Ex_{\substack{x_1, \dots, x_k \in S\\y_1, \dots, y_\ell \in T}} \prod_{\substack{i \in [k]\\j \in [\ell]}} A(x_i, y_j) \cdot \prod_{j \in [\ell]} A(x, y_j)}^{\frac{1}{k \ell}} \\
    &= \parens*{\Ex_{\substack{S \in \mathcal S\\T \in \mathcal T}} \Ex_{\substack{x_1, \dots, x_k \in S\\y_1, \dots, y_\ell \in T}} \prod_{\substack{i \in [k]\\j \in [\ell]}} A_x'(x_i, y_j)}^{\frac{1}{k \ell}} \\
    &= \parens*{\Ex_{\substack{S \in \mathcal S\\T \in \mathcal T}} \norm{A_x'[S, T]}_{U(k, \ell)}^{k \ell}}^{\frac{1}{k \ell}},
\intertext{and thus, by the previous \cref{lem:bounds-add-approx,lem:grid-norm-sampling},}
    &= \parens*{\norm{A_x'}_{U(k, \ell)}^{k \ell} \pm (2\epsilon k + 2\epsilon \ell + 2\delta)}^{\frac{1}{k\ell}} \\[1ex]
    &= \norm{A_x'}_{U(k, \ell)} \pm (2\epsilon k + 2\epsilon \ell + 2\delta)^{\frac{1}{k\ell}}.
\end{align*}
By the definitions of $A_x$ and $A_x'$, we have that $\norm{A_x'}_{U(k, \ell)}^{k \ell} = \deg_A(x)^\ell \cdot \norm{A_x}_{U(k, \ell)}^{k \ell}$ which implies that~\smash{$\norm{A_x'}_{U(k, \ell)} = \deg_A(x)^{\frac{1}{k}} \cdot \norm{A_x}_{U(k, \ell)}$} and thus
\begin{equation*}
    v_x = \norm{A_x}_{U(k, \ell)} \pm \frac{(2\epsilon k + 2\epsilon \ell + 2 \delta)^{\frac{1}{k \ell}}}{\deg_A(x)^{\frac{1}{k}}}.
\end{equation*}
To achieve the claimed bound $\alpha$ on the additive error, we choose \smash{$\epsilon = \delta = \alpha^{k \ell} / (2k + 2\ell + 2) = \alpha^{\Order(k \ell)}$}.
% $\deg_A(x) \cdot \norm{A_x}_{U(k, \ell)} = \norm{A_x'}_{U(k, \ell)}$. \shachar{I think this is an error. What we have is $\norm{A_x'}_{U(k, \ell)}^{k \ell} = \deg(A_x)^{\ell} \norm{A_x}_{U(k, \ell)}^{k \ell}$ which gives $\norm{A_x'}_{U(k, \ell)} = \deg(A_x)^{1/k} \norm{A_x}_{U(k, \ell)}$}
% And to achieve the claimed bound $\alpha$ on the additive error, we choose \smash{$\epsilon = \delta = \alpha^{k \ell} / (2k + 2\ell + 2) = \alpha^{\Order(k \ell)}$}.

We finally consider the running time. The precomputation takes time $n \poly(\epsilon^{-1}, \delta^{-1}, \log n)$ by \cref{lem:oblivious-sampler}. Computing the intermediate values $u_{T, y_1, \dots, y_\ell}$ takes time
\begin{equation*}
    \Order\parens*{\sum_{\substack{S \in \mathcal S\\T \in \mathcal T}} |S|^k \cdot |T|^\ell \cdot k \ell} = \Order\parens*{\sum_{\substack{S \in \mathcal S\\T \in \mathcal T}} \poly(\epsilon^{-1}, \delta^{-1})^{k + \ell}} = n^2 \cdot \poly(\epsilon^{-1}, \delta^{-1})^{k + \ell}.
\end{equation*}
%\shachar{where is the time dependence on $|S|^k$ and $|T|^{\ell}$?},
Similarly, computing the values $u_x$ then takes time~\smash{$\Order(|X| \cdot \sum_{T \in \mathcal T} |T|^\ell \cdot k \ell) = n^2 \poly(\epsilon^{-1}, \delta^{-1})^\ell$}. All contributions are bounded by $n^2 \cdot \alpha^{-\Order(k\ell(k + \ell))}$.
\end{proof}
% !TEX root = ../paper.tex
\section{Regularity Decompositions} \label{sec:decompositions}
In this section we establish the regularity decompositions (\cref{thm:edge-decomposition,thm:2-path-decomposition}). The structure of this section closely follows the outline from \cref{sec:overview}.

We start with the following lemma stating that any graph can either be made $\epsilon$-min-degree without loosing many nodes, or we can find a denser subgraph.

\begin{lemma}[Minimum Degree] \label{lem:min-degree}
Let $A \in \set{0, 1}^{X \times Y}$, and let $\epsilon, \gamma > 0$. There is an algorithm $\MinDegree(X, Y, A, \epsilon, \gamma)$ computing a set $X' \subseteq X$ of size $|X'| \geq \floor{(1 - \gamma) \cdot |X|}$ such that, writing $A' = A[X', Y]$, one of the following cases holds:
\begin{enumerate}
    \item $A'$ is $\epsilon$-min-degree and $\Ex[A'] \geq \Ex[A]$, or
    \item $\Ex[A'] \geq (1 + \gamma \epsilon) \Ex[A]$. %\shachar{this bound seems lossy, I don't see why $|X'|$ should depend on $\Ex[A]$ at all}
    % \shachar{might be easier to analyze the algorithm which stops if either $A'$ is $\epsilon$-min-degree or if we removed $\gamma |X|$ elements, and then show it must incur a density increment}
\end{enumerate}
The algorithm is deterministic and runs in time $\Order(|X| \, |Y|)$.
\end{lemma}
\begin{proof}
% \shachar{should we add pseudocode for this algorithm as well?}
Consider the following algorithm; for the pseudocode see \cref{alg:min-degree}. Initially, we assign~\makebox{$X' \gets X$} and $A' \gets A$. As long as there exists some $x \in X'$ with $\deg_{A'}(x) < (1 - \epsilon) \Ex[A']$, we remove $x$ by updating $X' \gets X' \setminus \set{x}$ and~\makebox{$A' \gets A[X', Y]$}. When this rule terminates we output the resulting set $X'$ (Case~1). If, however, we reach size $|X'| \le (1 - \gamma) \cdot |X|$, then we stop the algorithm prematurely and return the set~$X'$ from that stage of the algorithm (Case~2).

\begin{algorithm}[t]
\caption{Implements the algorithm from \cref{lem:min-degree}.} \label{alg:min-degree}
\begin{algorithmic}[1]
    \Procedure{\MinDegree}{$X, Y, A, \epsilon, \gamma$}
        \State Initialize $X' \gets X$ and $A' \gets A$
        \While{$\exists x \in X'$ with $\deg_{A'}(x) < (1 - \epsilon) \Ex[A']$}
            \State Update $X' \gets X' \setminus \set{x}$ and $A' \gets A[X', Y]$
            \If{$|X'| \leq (1 - \gamma) \cdot |X|$}
                \State\Return $X'$ (Case 2)
            \EndIf
        \EndWhile
        \State\Return $X'$ (Case 1)
    \EndProcedure
\end{algorithmic}
\end{algorithm}

\paragraph{Correctness of Case~1}
Suppose that the algorithm terminates in Case~1. It is clear that $A'$ is $\epsilon$-min-degree. Moreover, one can easily verify that $\Ex[A'] \geq \Ex[A]$ since we have only removed nodes with degree smaller than average. Finally, since the algorithm has not stopped in Case~2 before, we indeed have $|X'| \geq (1 - \gamma) \cdot |X|$.

% To prove that~\makebox{$|X'| \geq (1 - \gamma) \cdot |X|$}, we show by contradiction that assuming \smash{$|X'| < (1 - \gamma) \cdot |X'|$} would cause the algorithm to interrupt and terminate in Case~2. For convenience write $X'' := X \setminus X'$ and~\makebox{$A'' := A[X'', Y]$}. Note that:
% \begin{equation*}
%     \Ex[A] = \frac{|X'|}{|X|} \cdot \Ex[A'] + \frac{|X''|}{|X|} \cdot \Ex[A'']
% \end{equation*}
% Moreover, by construction we have that $\Ex[A''] \leq (1 - \epsilon) \Ex[A']$ (as the density $\Ex[A']$ can never decrease), and therefore:
% \begin{equation*}
%     \Ex[A] \leq \frac{|X'|}{|X|} \cdot \Ex[A'] + \frac{|X''|}{|X|} \cdot (1 - \epsilon) \Ex[A'] = \parens*{1 - \epsilon \cdot \frac{|X''|}{|X|}} \cdot \Ex[A'] \leq (1 - \gamma \epsilon) \cdot \Ex[A'].
% \end{equation*}
% It follows that $\Ex[A'] \geq (1 + \gamma \epsilon) \Ex[A]$, proving that the algorithm would indeed interrupt.

\paragraph{Correctness of Case~2}
Suppose now that the algorithm terminates in Case~2. We first argue that $\Ex[A'] \geq (1 + \gamma \epsilon) \Ex[A]$. To this end, we let $X'$ and $A'$ be as when the algorithm terminates. As this happens in the first iteration when $|X'| \leq (1 - \gamma) \cdot |X|$ we have that $|X'| = \floor{(1 - \gamma) \cdot |X|}$, and thus also $|X \setminus X'| \geq \gamma \cdot |X|$. Now let $A'' = A[X \setminus X', Y]$. Since the density of $A'$ only increases over the course of the algorithm, and since we only remove nodes with $\deg_{A'}(x) < (1 - \epsilon) \Ex[A']$, we have that $\Ex[A''] < (1 - \epsilon) \Ex[A']$. Therefore:
\begin{align*}
    \Ex[A] &= \frac{|X'|}{|X|} \cdot \Ex[A'] + \frac{|X \setminus X'|}{|X|} \cdot \Ex[A''] \\
    &\leq \frac{|X'|}{|X|} \cdot \Ex[A'] + \frac{|X \setminus X'|}{|X|} \cdot (1 - \epsilon) \Ex[A'] \\
    &= \parens*{1 - \epsilon \cdot \frac{|X \setminus X'|}{|X|}} \cdot \Ex[A'] \\[.6ex]
    &\leq (1 - \gamma \epsilon) \cdot \Ex[A'].
\end{align*}
Rearranging yields that $\Ex[A'] \geq (1 + \gamma \epsilon) \Ex[A]$ as claimed.

% To complete the correctness analysis, we prove that the size $|X'|$ is as claimed. First note that by construction we remove at most $\deg_{A'}(x) \cdot |Y|$ edges from $A'$ in each step. Using that~\makebox{$\deg_{A'}(x) \leq (1 - \epsilon) \Ex[A']$} and since $\Ex[A'] \leq (1 + \gamma \epsilon) \Ex[A]$ (as otherwise the algorithm would have terminated already), in each step we remove at most $(1 - \epsilon + \gamma \epsilon) \Ex[A] \cdot |Y|$ edges. In particular, the final edge set $A'$ still contains~\makebox{$\epsilon(1 - \gamma) |A|$} edges. It follows that
% \begin{equation*}
%     \frac{|X'|}{|X|} \geq \frac{|A'|}{|X| \, |Y|} \geq \frac{\epsilon (1 - \gamma) |A|}{|X| \, |Y|} = \epsilon (1 - \gamma) \Ex[A].
% \end{equation*}

\paragraph{Running Time}
It is easy to check that this algorithm can be implemented in time $\Order(|X| \, |Y|)$: We precompute the degrees of all nodes in $X$, and sort $X$ according to these degrees. Throughout we maintain the set $X$ and the size $|A'|$. In each step we can find in constant time a node $x \in X'$ with $\deg_{A'}(x) < (1 - \epsilon) \Ex[A']$ if it exists (namely the node in $X'$ with smallest degree).
\end{proof}

We remark that, while the previous lemma only guarantees the left-sided min-degree condition, it is equally possible to guarantee the condition on both sides. However, we never need this stronger statement in our upcoming proofs and therefore stick to this simpler version.

\subsection{\texorpdfstring{\boldmath$A$}{A}-Decomposition}
In this subsection we prove \cref{thm:edge-decomposition} (i.e., the decomposition of a \emph{single bipartite} graph into regular subgraphs). As outlined in \cref{sec:overview}, the proof consists of two steps: A method to find \emph{good rectangles} via density increments (see \cref{lem:good-rect}), and a decomposition via density decrements that repeatedly remove good rectangles (see \cref{thm:edge-decomposition}). 

\begin{lemma}[Finding a Good Rectangle] \label{lem:good-rect}
Let $A \in \set{0, 1}^{X \times Y}$, let~\makebox{$\epsilon \in (0, 1)$} and $d \geq 1$, and assume that $\Ex[A] \geq 2^{-d}$. There is an algorithm $\RegRect(X, Y, A, \epsilon, d)$ computing~\makebox{$X^* \subseteq X, Y^* \subseteq Y$} such that:
\begin{enumerate}
    \item Let $A^* = A[X^*, Y^*]$. Then $A^*$ is $(\epsilon, 2, d)$-regular and $\epsilon$-min-degree.
    \item $\Ex[A^*] \geq \Ex[A]$.
    \item $|X^*| \, |Y^*| \geq \exp(-d^3 \poly(\epsilon^{-1})) \cdot |X| \, |Y|$.
\end{enumerate}
The algorithm is deterministic and runs in time $n^2 \cdot \exp(d^3 \poly(\epsilon^{-1}))$ (where $n = |X| + |Y|$).
\end{lemma}
\begin{proof}
We start with the description of the algorithm; see \cref{alg:reg-rect} for the pseudocode. We first apply \cref{lem:min-degree} to compute~\makebox{$X' \subseteq X$} and $A' \gets A[X', Y]$ (with parameters $\epsilon$ and $\gamma \gets \frac12$, \cref{alg:reg-rect:line:min-degree}). The lemma guarantees one of two cases: Either~$A'$ is $\epsilon$-min-degree, or $\Ex[A'] \geq (1 + \frac{\epsilon}{2}) \Ex[A]$. In the latter case we recurse on the subinstance~$(X', Y, A')$ to find a good rectangle (\cref{alg:reg-rect:line:min-degree-dense,alg:reg-rect:line:min-degree-recur}). In the former case we continue and apply \cref{thm:sifting} to $(X', Y, A')$ (with parameters $\epsilon$ and $d$). If the theorem returns a rectangle~\makebox{$X'' \times Y'' \subseteq X' \times Y$} with~\makebox{$\Ex[A[X'', Y'']] \geq (1 + \frac{\epsilon}{2}) \Ex[A]$}, then we recurse on the subinstance $(X'', Y'', A[X'', Y''])$ (\cref{alg:reg-rect:line:sifting,alg:reg-rect:line:sifting-recur}). Otherwise, \cref{thm:sifting} guarantees that $A'$ is $(\epsilon, 2, d)$-regular. In this case we finally return $X^* \gets X', Y^* \gets Y$ (\cref{alg:reg-rect:line:return}).
% \shachar{is it possible to have the algorithms appear in the appropriate place, instead of top the page which is the middle of the previous step?}

\begin{algorithm}[t]
\caption{Implements the algorithm from \cref{lem:good-rect}.} \label{alg:reg-rect}
\begin{algorithmic}[1]
    \Procedure{\RegRect}{$X, Y, A, \epsilon, d$}
        \State Compute $X' \gets \MinDegree(X, Y, A, \epsilon, \frac{1}{2})$ and $A' \gets A[X', Y]$ \label{alg:reg-rect:line:min-degree}
        \If{$\Ex[A'] \geq (1 + \frac{\epsilon}{2}) \Ex[A]$} \label{alg:reg-rect:line:min-degree-dense}
            \State\Return $\RegRect(X', Y, A', \epsilon, d)$ \label{alg:reg-rect:line:min-degree-recur}
        \EndIf
        \If{$\Sift(X', Y, A', \epsilon, d)$ returns a denser rectangle $X'' \times Y'' \subseteq X' \times Y$} \label{alg:reg-rect:line:sifting}
            \State\Return $\RegRect(X'', Y'', A[X'', Y''], \epsilon, d)$ \label{alg:reg-rect:line:sifting-recur}
        \EndIf
        \State\Return $X', Y$ \label{alg:reg-rect:line:return}
    \EndProcedure
\end{algorithmic}
\end{algorithm}

The correctness of Property~1 is clear: We either recurse, or return the submatrix $A'$ that is $\epsilon$-min-degree and $(\epsilon, 2, d)$-regular. Moreover, it is easy to prove that Property~2 holds: \cref{lem:min-degree} states that $\Ex[A'] \geq \Ex[A]$ in either case; and if the algorithm recurses, then the density even strictly increases. It remains to prove Property~3 and to analyze the running time of this algorithm.

\paragraph{Correctness of Property~3}
First assume that the algorithm does not recurse and returns~$X', Y$. In this case, \cref{lem:min-degree} guarantees that $|X'| \geq \frac{1}{2} \cdot |X|$, and thus $|X^*| \, |Y^*| \geq \frac12 \cdot |X| \, |Y|$. Next, consider the recursive cases. By \cref{lem:min-degree,thm:sifting}, in both cases we recurse on a rectangle of size at least $\min\set{\frac{1}{2}, \frac{\epsilon}{16} \cdot \Ex[A]^{2d}} \cdot |X| \, |Y| = \frac{\epsilon}{16} \cdot \Ex[A]^{2d} \cdot |X| \, |Y|$ and with density at least~\makebox{$(1 + \frac{\epsilon}{2}) \Ex[A]$}. It follows by induction that the algorithm returns a rectangle of size
\begin{equation*}
    |X^*| \, |Y^*| \geq \tfrac{1}{2} \cdot \parens*{\tfrac{\epsilon}{16} \cdot \Ex[A]^{2d}}^{\log_{1+\frac{\epsilon}{2}}(\Ex[A]^{-1})} \cdot |X| \, |Y| \geq \exp(-d^3 \poly(\epsilon^{-1})) \cdot |X| \, |Y|,
\end{equation*}
where in the latter bound we used that $\log(1 + \epsilon) \geq \epsilon$ for all $\epsilon \in (0, 1)$, and that initially $\Ex[A] \geq 2^{-d}$.

\paragraph{Running Time}
The algorithm reaches recursion depth at most \smash{$\log_{1+\frac{\epsilon}{2}}(\Ex[A]^{-1}) = \Order(d / \epsilon)$}, which causes a negligible overhead in the running time. The call to $\MinDegree$ (\cref{lem:min-degree}) takes time~$\Order(n^2)$, and the call to $\Sift$ (\cref{thm:sifting}) takes time $n^2 \cdot (\epsilon\Ex[A])^{-\Order(d^2)} = n^2 \cdot \exp(d^3 \poly(\epsilon^{-1}))$. All in all, the running time is $n^2 \cdot \exp(d^3 \poly(\epsilon^{-1}))$ as claimed.
\end{proof}

\thmedgedecomposition*
\begin{proof}
We start with the description of the algorithm; see \cref{alg:edge-decomposition} for the pseudocode. As a first step, we compute the density~$\Ex[A]$ and test whether $\Ex[A] \leq 2^{-d}$. In this case we can immediately stop and return the trivial decomposition~$\set{(X, Y, A)}$ (\cref{alg:edge-decomposition:line:sparse,alg:edge-decomposition:line:sparse-return}). Otherwise, we can apply \cref{lem:good-rect} to compute a good rectangle~\makebox{$X^* \times Y^* \subseteq X \times Y$}. We then take $(X^*, Y^*, A[X^*, Y^*])$ as one piece of the decomposition, and recurse on the remaining graph where we remove all edges in~\makebox{$X^* \times Y^*$} (\cref{alg:edge-decomposition:line:reg-rect,alg:edge-decomposition:line:recur}; in the pseudocode by slight abuse of notation we denote the remaining graph by $A - A[X^*, Y^*]$).

\begin{algorithm}[t]
\caption{Implements the algorithm from \cref{thm:edge-decomposition}.} \label{alg:edge-decomposition}
\begin{algorithmic}[1]
    \Procedure{\ADecomposition}{$X, Y, A, \epsilon, d$}
        \If{$\Ex[A] \leq 2^{-d}$} \label{alg:edge-decomposition:line:sparse}
            \State\Return $\set{(X, Y, A)}$ \label{alg:edge-decomposition:line:sparse-return}
        \EndIf
        \State Compute $X^*, Y^* \gets \RegRect(X, Y, A, \epsilon, d)$ \label{alg:edge-decomposition:line:reg-rect}
        \State\Return $\set{(X^*, Y^*, A[X^*, Y^*])} \cup \ADecomposition(X, Y, A - A[X^*, Y^*], \epsilon, d)$ \label{alg:edge-decomposition:line:recur}
    \EndProcedure
\end{algorithmic}
\end{algorithm}

It is easy to see that Property~1 holds: We either cover $A$ entirely in the base case, or we cover some part $X^* \times Y^*$ and remove that part in the recursive call. Moreover, Property~2 easily follows from the guarantee of \cref{lem:good-rect}. It remains to prove Properties~3 and 4, and to analyze the running time.

\paragraph{Correctness of Property~3}
For a collection of tuples $\mathcal S$ as returned by the algorithm, let us define
\begin{equation*}
    C(\mathcal S) = \sum_{(X', Y', A') \in \mathcal S} |X'| \, |Y'|.
\end{equation*}
Our goal is to prove that $C(\mathcal S) \leq (d + 2) \cdot |X| \, |Y|$, for any input $(X, Y, A)$, where $\mathcal S$ is the set returned by the algorithm. This is clear whenever $\Ex[A] \leq 2^{-d}$ (as then the algorithm returns the trivial partition $\set{(X, Y, A)}$). By induction we prove that whenever $\Ex[A] \geq 2^{-d}$ then
\begin{equation*}
    C(\mathcal S) \leq (d + 2 - \log(\Ex[A]^{-1})) \cdot |X| \, |Y|.
\end{equation*}
Let $\mathcal S^*$ denote the output of the recursive call $\ADecomposition(X, Y, A - A[X^*, Y^*], \epsilon, d)$, and let $\delta^*$ denote the density of $A - A[X^*, Y^*]$. Clearly, $C(\mathcal S) = |X^*| \, |Y^*| + C(\mathcal S^*)$. If $\delta^* \leq 2^{-d}$, then $C(\mathcal S^*) \leq |X| \, |Y|$ and thus $C(\mathcal S) \leq 2 |X| \, |Y|$ as claimed. Otherwise, recall that \cref{lem:good-rect} guarantees that the density of the rectangle $X^* \times Y^*$ does not decrease, $\Ex[A[X^*, Y^*]] \geq \Ex[A]$. It follows that \smash{$\delta^* \geq \Ex[A] - \Ex[A] \cdot \frac{|X^*| \, |Y^*|}{|X| \, |Y|}$}, and thus by induction:
\begin{align*}
    C(\mathcal S) &= |X^*| \, |Y^*| + C(\mathcal S^*) \\
    &\leq |X^*| \, |Y^*| + (d + 2 - \log((\delta^*)^{-1}) \cdot |X| \, |Y| \\
    &\leq |X^*| \, |Y^*| + (d + 2 - \log(\Ex[A]^{-1}) - \log((1 - \tfrac{|X^*| \, |Y^*|}{|X| \, |Y|})^{-1})) \cdot |X| \, |Y| \\
    &\leq |X^*| \, |Y^*| + (d + 2 - \log(\Ex[A]^{-1}) - \tfrac{|X^*| \, |Y^*|}{|X| \, |Y|}) \cdot |X| \, |Y| \\
    &= (d + 2 - \log(\Ex[A]^{-1})) \cdot |X| \, |Y|.
\end{align*}

\paragraph{Correctness of Property~4}
Note that $L$ is the recursion depth of the algorithm. To prove that~$L$ is bounded as claimed, we argue that with every recursive call the density of the matrix decreases. Specifically, each output $X^*, Y^*$ has size at least $|X^*| \, |Y^*| \geq \exp(-d^3 \poly(\epsilon^{-1})) \cdot |X| \cdot |Y|$ by \cref{lem:good-rect}, and the density of the submatrix $A^* = A[X^*, Y^*]$ satisfies~\smash{$\Ex[A^*] \geq \Ex[A] \geq 2^{-d}$} (as otherwise, if $\Ex[A] \leq 2^{-d}$, the algorithm had terminated already). %\shachar{not true, the matrix $A$ becomes sparser in each step}
Thus, each recursive call reduces the density of the graph by~\smash{$\Ex[A^*] \cdot \frac{|X^*| \, |Y^*|}{|X| \, |Y|} \geq 2^{-d} \cdot \exp(-d^3 \poly(\epsilon^{-1})) = \exp(-d^3 \poly(\epsilon^{-1}))$}, %\shachar{this is not true, because as I wrote the matrix $A$ becomes sparser. What is true is that at each step we remove a $\exp(-d^3 \poly(\epsilon^{-1}))$ fraction of edges from the current matrix}
and the algorithm necessarily terminates after $L \leq \exp(d^3 \poly(\epsilon^{-1}))$ recursive calls.

\paragraph{Running Time}
We have already bounded the recursion depth in the previous paragraph, so focus on a single execution. The dominant cost is the call to $\RegRect$ (\cref{lem:good-rect}) which takes time $n^2 \cdot \exp(d^3 \poly(\epsilon^{-1}))$. The total time is $L \cdot n^2 \cdot \exp(d^3 \poly(\epsilon^{-1})) = n^2 \cdot \exp(d^3 \poly(\epsilon^{-1}))$.
\end{proof}

\subsection{\texorpdfstring{\boldmath$AB$}{AB}-Decomposition}
We finally turn to the proof of \cref{thm:2-path-decomposition}. The outline, as discussed in \cref{sec:overview}, follows the previous subsection on a high-level (but differs in many more difficult technical aspects). We first devise a method to find \emph{good cube} via density increments (see \cref{lem:good-cube}), and then derive the decomposition via density decrements that repeatedly remove good cubes (see \cref{thm:2-path-decomposition}). 

\begin{lemma}[Finding a Good Cube] \label{lem:good-cube}
Let $A \in \set{0, 1}^{X \times Y}, B \in \set{0, 1}^{Y \times Z}$, let $\epsilon \in (0, 1), \gamma \in (0, \frac12)$ and~\makebox{$d \geq 1$} and assume that $\Ex[B] \geq 2^{-d}$. %\shachar{why assume just for $B$ and not for $A$?}\nick{Because don't need it for $A$ here; we get this guarantee by calling the $A$-decomposition. We only need it for $B$ because we want to use density increments on $B$.}
There is an algorithm $\RegCube(X, Y, Z, A, B, \epsilon, \gamma, d)$ computing sets~$Y^* \subseteq Y$, $Z^* \subseteq Z$ and $\set{(X_\ell, Y_\ell, Z_\ell, A_\ell)}_{\ell=1}^L$ with $X_\ell \subseteq X$,~$Y_\ell \subseteq Y^*$,~\makebox{$Z_\ell \subseteq Z^*$} and~$A_\ell \in \set{0, 1}^{X_\ell \times Y_\ell}$ such that
\begin{enumerate}
    % \item $\set{(x, y) \in X \times Y^* : A(x, y) = 1} = \bigcup_{\ell=1}^L \set{(x, y) \in X_\ell \times Y_\ell : A_\ell(x, y) = 1}$.
    \item $A[X, Y^*] = \sum_{\ell=1}^L A_\ell$.
    \item For all $\ell \in [L]$, writing $B_\ell = B[Y_\ell, Z_\ell]$:
    \begin{itemize}
        \item $\Ex[A_\ell] \leq 2^{-d}$, or%\shachar{don't you also need to add the case that $B_{\ell}$ is sparse? otherwise we cannot really use the conclusion in the next bullet}\nick{Actually, here we guarantee the stronger property that $B_\ell$ is regular and min-degree (and not sparse).}
        \item $A_\ell$ and $B_\ell^T$ are both $(\epsilon, 2, d)$-regular and $\epsilon$-min-degree. 
    \end{itemize}
    \item $\Ex[B[Y^*, Z^*]] \geq \Ex[B]$.
    \item $\sum_{\ell=1}^L |X_\ell| \, |Y_\ell| \, |Z_\ell| \leq (d + 2) \cdot |X| \, |Y^*| \, |Z^*|$
    \item $\sum_{\ell=1}^L |X_\ell| \, |Y_\ell| \, |Z^* \setminus Z_\ell| \leq \gamma (d + 2) \cdot |X| \, |Y^*| \, |Z^*|$.
    \item $|Y^*| \, |Z^*| \geq \exp(-d^4 \gamma^{-1} \poly(\epsilon^{-1})) \cdot |Y| \, |Z|$.
    \item $L \leq \exp(d^3 \poly(\epsilon^{-1}))$.
\end{enumerate}
The algorithm is deterministic and runs in time $n^2 \cdot \gamma^{-1} \exp(d^3 \poly(\epsilon^{-1}))$ (for $n = |X| + |Y| + |Z|$).
\end{lemma}
\begin{proof}
Let us start with a description of the algorithm; for the pseudocode see \cref{alg:reg-cube}. We first call $\MinDegree(Y, Z, B, \frac{\epsilon\gamma}{2}, \frac12)$ (\cref{lem:min-degree}) to compute some $Y' \subseteq Y$. Write $A' \gets A[X, Y']$ and $B' \gets B[Y', Z]$. \cref{lem:min-degree} states that either $B'$ is $\frac{\epsilon\gamma}{2}$-min-degree, or~\makebox{$\Ex[B'] \geq (1 + \frac{\epsilon\gamma}{4}) \Ex[B]$}. In the latter case we simply recurse on the subinstance induced by $X, Y', Z$ (\crefrange{alg:reg-cube:line:min-degree-left}{alg:reg-cube:line:min-degree-left-recur}).

Next, we run \cref{thm:edge-decomposition} to compute an edge decomposition $\set{(X_\ell, Y_\ell, A_\ell)}_{\ell=1}^L$ of $(X, Y', A')$ (\cref{alg:reg-cube:line:edge-decomposition}). The hope is that, for all pieces $\ell \in [L]$, the induced graphs $B[Y_\ell, Z]$ also satisfy the min-degree and regularity conditions, in which case we could return the decomposition without changes. To ensure both, we enumerate each $\ell \in [L]$ (\cref{alg:reg-cube:line:loop}). By calling $\MinDegree(Z, Y_\ell, B^T[Z, Y_\ell], \epsilon, \gamma)$ (\cref{lem:min-degree}) we compute a set $Z_\ell \subseteq Z$ such that, writing $B_\ell \gets B[Y_\ell, Z_\ell]$, either $B_\ell^T$ is $\epsilon$-min-degree or $B_\ell$ has increased density. The former case is exactly the desired min-degree condition, and in the latter case we again recurse on the subinstance induced by $X, Y_\ell, Z_\ell$ (\crefrange{alg:reg-cube:line:min-degree-right}{alg:reg-cube:line:min-degree-right-recur}). It remains to ensure regularity. To this end we call $\Sift(Z_\ell, Y_\ell, B_\ell^T, \epsilon, 2, d)$ (\cref{thm:sifting}), which either certifies that $B_\ell^T$ is $(\epsilon, 2, d)$-regular, or finds $Z'' \subseteq Z_\ell$ and $Y'' \subseteq Y_\ell$ such that the density of the induced subgraph $B_\ell[Y'', Z'']$ increases. In the latter case, again, we recurse (\cref{alg:reg-cube:line:sifting,alg:reg-cube:line:sifting-recur}).

If after all these tests the algorithm has not recursed, we finally return $Y^* \gets Y'$, $Z^* \gets Z$ and the collection $\set{(X_\ell, Y_\ell, Z_\ell, A_\ell)}_{\ell=1}^L$ (\cref{alg:reg-cube:line:return}).

\begin{algorithm}[t]
\caption{Implements the algorithm from \cref{lem:good-cube}.} \label{alg:reg-cube}
\begin{algorithmic}[1]
    \Procedure{\RegCube}{$X, Y, Z, A, B, \epsilon, \gamma, d$}
        \State Compute $Y' \gets \MinDegree(Y, Z, B, \frac{\epsilon \gamma}{2}, \frac12)$ and let $A' \gets A[X, Y']$ and $B' \gets B[Y', Z]$ \label{alg:reg-cube:line:min-degree-left}
        \If{$\Ex[B'] \geq (1 + \frac{\epsilon \gamma}{4}) \Ex[B]$} \label{alg:reg-cube:line:min-degree-left-dense}
            \State\Return $\RegCube(X, Y', Z, A', B', \epsilon, \gamma, d)$ \label{alg:reg-cube:line:min-degree-left-recur}
        \EndIf
        \State Compute $\set{(X_\ell, Y_\ell, A_\ell)}_{\ell=1}^L \gets \ADecomposition(X, Y', A', \epsilon, d)$ \label{alg:reg-cube:line:edge-decomposition}
        \ForEach{$\ell \in [L]$} \label{alg:reg-cube:line:loop}
            \State Compute $Z_\ell \gets \MinDegree(Z, Y_\ell, B^T[Z, Y_\ell], \epsilon, \gamma)$ and let $B_\ell \gets B[Y_\ell, Z_\ell]$ \label{alg:reg-cube:line:min-degree-right}
            \If{$\Ex[B_\ell] \geq (1 + \epsilon\gamma) \Ex[B[Y_\ell, Z]]$} \label{alg:reg-cube:line:min-degree-right-dense}
                \State\Return $\RegCube(X, Y_\ell, Z_\ell, A[X, Y_{\ell}], B_\ell, \epsilon, \gamma, d)$ \label{alg:reg-cube:line:min-degree-right-recur} %\shachar{replaced $A_{\ell}$ with $A[X, Y_{\ell}]$, please verify}
            \EndIf
            \If{$\Sift(Z_\ell, Y_\ell, B_\ell^T, \epsilon, 2, d)$ returns a denser rectangle $Z'' \times Y'' \subseteq Z_\ell \times Y_\ell$} \label{alg:reg-cube:line:sifting}
                \State\Return $\RegCube(X, Y'', Z'', A[X, Y''], B[Y'', Z''], \epsilon, \gamma, d)$ \label{alg:reg-cube:line:sifting-recur}
            \EndIf
        \EndForEach
        \State\Return $Y', Z, \set{(X_\ell, Y_\ell, Z_\ell, A_\ell)}_{\ell=1}^L$ \label{alg:reg-cube:line:return}
    \EndProcedure
\end{algorithmic}
\end{algorithm}

\medskip
Some properties of the algorithm are easy to prove. For instance, Properties~1 and~7 follow immediately from \cref{thm:edge-decomposition}. Property~2 is easy to prove as well: \cref{thm:edge-decomposition} implies that for each $\ell \in [L]$, we have that $\Ex[A_\ell] \leq 2^{-d}$ or that $A_\ell$ is $\epsilon$-min-degree and $(\epsilon, 2, d)$-regular. In addition, the algorithm only terminates and returns an output after certifying that, for all $\ell \in [L]$, $B_\ell^T$ is $\epsilon$-min-degree and $(\epsilon, 2, d)$-regular. The other properties require more work. We start with the following claim:

\begin{claim*}
The algorithm only recurses on subgraphs of $B$ with density at least $(1 + \frac{\epsilon\gamma}{4}) \Ex[B]$.
\end{claim*}
\begin{proof}
If the algorithm recurses in \cref{alg:reg-cube:line:min-degree-left-dense,alg:reg-cube:line:min-degree-left-recur}, then the claim is immediate. After passing \cref{alg:reg-cube:line:min-degree-left-recur}, \cref{lem:min-degree} guarantees that the graph $B'$ is $\frac{\epsilon\gamma}{2}$-min-degree. From this min-degree condition we know that, for any set~\makebox{$Y'' \subseteq Y'$}, the subgraph $B[Y'', Z]$ has density at least $(1 - \frac{\epsilon\gamma}{2}) \Ex[B]$. In particular, if the algorithm recurses in \cref{alg:reg-cube:line:min-degree-left-recur} then we recurse on a subgraph of density
\begin{equation*}
    \Ex[B_\ell] \geq (1 + \epsilon \gamma) \Ex[B[Y_\ell, Z]] \geq (1 + \epsilon \gamma) (1 - \tfrac{\epsilon\gamma}{2}) \Ex[B] \geq (1 + \tfrac{\epsilon\gamma}{4}) \Ex[B];
\end{equation*}
here in the last step we used that $\gamma \in (0, \frac12)$ and $\epsilon \in (0, 1)$. Similarly, if the algorithm recurses in \cref{alg:reg-cube:line:sifting-recur} then the density is at least $(1 + \tfrac{\epsilon}{2}) \Ex[B_\ell] \geq (1 + \tfrac{\epsilon}{2}) \Ex[B[Y_\ell, Z]] \geq (1 + \tfrac{\epsilon\gamma}{4}) \Ex[B]$.
\end{proof}

\paragraph{Correctness of Property 3}
With the previous claim in mind we can easily conclude that Property~3 holds. If the algorithm terminates without recurring, then $\Ex[B[Y^*, Z^*]] \geq \Ex[B]$ follows directly from \cref{lem:min-degree}. If the algorithm recurses, then the statement follows by induction using that the density never decreases.

\paragraph{Correctness of Properties~4 and~5}
Observe that Property~4 is immediate by \cref{thm:edge-decomposition}. Furthermore, to prove Property~5 it suffices to check that $|Z^* \setminus Z_\ell| \leq \gamma |Z^*|$ for all $\ell \in [L]$. This indeed holds by \cref{lem:min-degree}.

\paragraph{Correctness of Property~6}
First note that in the base case, if the algorithm terminates without recurring, then $Y^* = Y'$, $Z^* = Z$ and thus $|Y^*| \, |Z^*| \geq \frac12 \cdot |Y| \, |Z|$. Next consider the recursive cases. In \cref{alg:reg-cube:line:min-degree-left-recur} we possibly recurse on the subgraph $B[Y', Z]$, where \smash{$|Y'| \geq \frac{1}{2} \cdot |Y|$}. In \cref{alg:reg-cube:line:min-degree-right-recur} we possibly recurse on the subgraph $B[Y_\ell, Z_\ell]$, where $|Y_\ell| \geq \exp(-d^3 \poly(\epsilon^{-1})) \cdot |Y'|$ (by \cref{thm:edge-decomposition}) and~\smash{$|Z_\ell| \geq (1 - \gamma) \cdot |Z| \geq \frac{1}{2} \cdot |Z|$} (by \cref{lem:min-degree}). Finally, in \cref{alg:reg-cube:line:sifting-recur} we possibly recurse on a subgraph~$B[Y'', Z'']$ of size at least $|Y''| \, |Z''| \geq \frac{\epsilon}{16} \cdot 2^{-2d^2} \cdot |Y_\ell| \, |Z_\ell|$ (by \cref{thm:sifting}). In either case, the size of the subgraph is at least \smash{$\exp(-d^3 \poly(\epsilon^{-1})) \cdot |X| \, |Y|$}. Recall further that by the claim, the density increases multiplicatively by $1 + \frac{\epsilon\gamma}{4}$ with every recursive call. Therefore, and since the initial density is $\Ex[B] \geq 2^{-d}$, the algorithm returns sets $Y^*, Z^*$ with size
\begin{equation*}
    \exp(-d^3 \poly(\epsilon^{-1}))^{\log_{1 + \frac{\epsilon\gamma}{4}}(\Ex[B]^{-1})} \cdot |Y| \, |Z| \geq \exp(-d^4 \gamma^{-1} \poly(\epsilon^{-1})) \cdot |Y| \, |Z|.
\end{equation*}
% \shachar{there is a $d^4$ factor here, but in the theorem statement it says $d^3$. We need to check if this error propagates to the followup statements as well}
Here we used that $\log(1 + \epsilon) \geq \epsilon$ for all $\epsilon \in (0, 1)$.

\paragraph{Running Time}
We finally analyze the running time of the algorithm. As argued before, the recursion depth is bounded by $\Order(d \epsilon^{-1} \gamma^{-1})$. In each execution of \cref{alg:reg-cube} we call $\ADecomposition$ once which takes time $n^2 \cdot \exp(d^3 \poly(\epsilon^{-1}))$. In addition we call $\Sift$ $L$ times taking time~$L \cdot n^2 \cdot \exp(d^2 \poly(\epsilon^{-1}))$. The total running time becomes $n^2 \cdot \gamma^{-1} \exp(d^3 \poly(\epsilon^{-1}))$.
\end{proof}

\thmtwopathdecomposition*
\begin{proof}
We start with the description of the algorithm; see \cref{alg:2-path-decomposition} for the pseudocode. Throughout we assume an additional input parameter $0 \leq h \leq d$ which acts somewhat as the recursion depth of the algorithm. For the initial call, we set $h = 0$.

The algorithm has two bases cases. If $B$ is sufficiently sparse, $\Ex[B] \leq 2^{-d}$, then we return the trivial decomposition $\set{(X, Y, Z, A, B)}$ (\cref{alg:2-path-decomposition:line:sparse,alg:2-path-decomposition:line:sparse-return}). Moreover, if the algorithm has reached recursion depth $h = d$, then we return a trivial decomposition of size at most $2^d$. Specifically, we split $B$ into submatrices $B_1, \dots, B_{2^d} \in \set{0, 1}^{Y \times Z}$ each of density at most $2^{-d}$ and return the decomposition $\set{(X, Y, Z, A, B_i)}_{i=1}^{2^d}$ (\crefrange{alg:2-path-decomposition:line:depth}{alg:2-path-decomposition:line:depth-return}).

Otherwise, run \smash{$\RegCube(X, Y, Z, A, B, \epsilon, \gamma, d)$} (\cref{lem:good-cube}) with the parameter \smash{$\gamma = \frac{1}{2(d+2)^2}$}. This output consists of sets $Y^* \subseteq Y$, $Z^* \subseteq Z$ and a set $\set{(X_\ell, Y_\ell, Z_\ell, A_\ell)}_{\ell=1}^L$ (\cref{alg:2-path-decomposition:line:reg-cube}). For each~\makebox{$\ell \in [L]$} we write $B_\ell \gets B[Y_\ell, Z_\ell]$ and $B_\ell' \gets B[Y_\ell, Z^* \setminus Z_\ell]$. We then recursively compute, for each $\ell \in [L]$, the decomposition $\mathcal S_\ell$ of $(X_\ell, Y_\ell, Z^* \setminus Z_\ell, A_\ell, B_\ell')$ (\crefrange{alg:2-path-decomposition:line:loop}{alg:2-path-decomposition:line:loop-recur}). Moreover, we recursively compute the decomposition $\mathcal S^*$ of $(X, Y, Z, A, B - B[Y^*, Z^*])$ (\cref{alg:2-path-decomposition:line:recur}; here, we denote by $B - B[Y^*, Z^*]$ the matrix obtained from $B$ after zeroing out the entries in $B[Y^*, Z^*]$). Finally, we return $\bigcup_{\ell=1}^L \set{(X_\ell, Y_\ell, Z_\ell, A_\ell, B_\ell)} \cup \bigcup_{\ell=1}^L \mathcal S_\ell \cup \mathcal S^*$ (\cref{alg:2-path-decomposition:line:return}).

\begin{algorithm}[t]
\caption{Implements the algorithm from \cref{thm:2-path-decomposition}.} \label{alg:2-path-decomposition}
\begin{algorithmic}[1]
    \Procedure{\ABDecompositionp}{$X, Y, Z, A, B, \epsilon, d, h$}
        \If{$\Ex[B] \leq 2^{-d}$} \label{alg:2-path-decomposition:line:sparse}
            \State\Return $\set{(X, Y, Z, A, B)}$ \label{alg:2-path-decomposition:line:sparse-return}
        \EndIf
        \If{$h = d$} \label{alg:2-path-decomposition:line:depth}
            \State Arbitrarily partition $B$ into submatrices $B_1, \dots, B_{2^d} \in \set{0, 1}^{Y \times Z}$ of density at most $2^{-d}$\label{alg:2-path-decomposition:line:depth-trivial-decomp}
            \State\Return \smash{$\set{(X, Y, Z, A, B_i)}_{i=1}^{2^d}$}\label{alg:2-path-decomposition:line:depth-return}
        \EndIf
        \State Compute \smash{$Y^*, Z^*, \set{(X_\ell, Y_\ell, Z_\ell, A_\ell)}_{\ell=1}^L \gets \RegCube(X, Y, Z, A, B, \epsilon, \frac{1}{2(d+2)^2}, d)$} \label{alg:2-path-decomposition:line:reg-cube}
        \ForEach{$\ell \in [L]$} \label{alg:2-path-decomposition:line:loop}
            \State Let $B_\ell \gets B[Y_\ell, Z_\ell]$ and $B_\ell' \gets B[Y_\ell, Z^* \setminus Z_\ell]$ \label{alg:2-path-decomposition:line:loop-names}
            \State Compute $\mathcal S_\ell \gets \ABDecompositionp(X_\ell, Y_\ell, Z^* \setminus Z_\ell, A_\ell, B_\ell', \epsilon, d, h + 1)$ \label{alg:2-path-decomposition:line:loop-recur}
        \EndForEach
        \State Compute $\mathcal S^* \gets \ABDecompositionp(X, Y, Z, A, B - B[Y^*, Z^*], \epsilon, d, h)$ \label{alg:2-path-decomposition:line:recur}
        \State\Return $\bigcup_{\ell=1}^L \set{(X_\ell, Y_\ell, Z_\ell, A_\ell, B_\ell)} \cup \bigcup_{\ell=1}^L \mathcal S_\ell \cup \mathcal S^*$ \label{alg:2-path-decomposition:line:return}
    \EndProcedure

    \bigskip
    \Procedure{\ABDecomposition}{$X, Y, Z, A, B, \epsilon, d$}
        \State\Return $\ABDecompositionp(X, Y, Z, A, B, \epsilon, d, 0)$
    \EndProcedure
\end{algorithmic}
\end{algorithm}

\paragraph{Correctness of Property 1}
For a set $\mathcal S$ as returned by the algorithm, we write
\begin{equation*}
    \Sigma(\mathcal S) = \sum_{(X', Y', Z', A', B') \in \mathcal S} A' B',
\end{equation*}
where, as in the theorem statement, we interpret each term $A' B'$ in the sum as an $X \times Z$-matrix by extending with zeros. The goal is to prove that $\Sigma(\mathcal S) = A B$, where $\mathcal S$ is the set returned by the algorithm on input $(X, Y, Z, A, B)$. This is clear in both base cases. So assume that the algorithm recurses. Then by induction:
\begin{align*}
    \Sigma(\mathcal S) &= \sum_{\ell=1}^L A_\ell B_\ell + \sum_{\ell=1}^L \Sigma(\mathcal S_\ell) + \Sigma(\mathcal S^*) \\
    &= \sum_{\ell=1}^L A_\ell B_\ell + \sum_{\ell=1}^L A_\ell B_\ell' + A B - A[X, Y^*] B[Y^*, Z^*] \\
    &= \sum_{\ell=1}^L A_\ell B[Y_\ell, Z^*] + A B - A[X, Y^*] B[Y^*, Z^*] \\[1ex]
    &= A[X, Y^*] B[Y^*, Z^*] + A B - A[X, Y^*] B[Y^*, Z^*] \\[1ex]
    &= AB;
\end{align*}
here, in the second-to-last step we have applied Property~1 of \cref{lem:good-cube}.
% Let $P(X, Y, Z, A, B)$ denote the set of 2-paths in the 3-layered graph $(X, Y, Z, A, B)$, that is,
% \begin{equation*}
%     P(X, Y, Z, A, B) = \set{(x, y, z) \in X \times Y \times Z : A(x, y) = B(y, z) = 1}.
% \end{equation*}
% To prove Property~1 we show the slightly stronger claim that
% \begin{equation*}
%     P(X, Y, Z, A, B) = \biguplus_{k=1}^K P(X_k, Y_k, Z_k, A_k, B_k).
% \end{equation*}
% The statement is clear in both base cases. In the recursive case the algorithm
% \begin{align*}
%     P(X, Y, Z, A, B) &= \biguplus_{\ell=1}^L P(X_\ell, Y_\ell, Z_\ell, A_\ell, B_\ell) \uplus \biguplus_{\ell=1}^L P(X_\ell, Y_\ell, Z^* \setminus Z_\ell, A_\ell, B_\ell') \uplus P(X, Y, Z, A, B - B[Y^*, Z^*]) \\
%     &= \biguplus_{\ell=1}^L P(X_\ell, Y_\ell, Z^*, A_\ell, B[Y_\ell, Z^*]) \uplus P(X, Y, Z, A, B - B[Y^*, Z^*]) \\
%     &= P(X, Y^*, Z^*, A[X, Y^*], B[Y^*, B^*]) \uplus P(X, Y, Z, A, B - B[Y^*, Z^*])
% \end{align*}
%
% We prove the slightly stronger claim that 
% \todo{Let $\Sigma(\mathcal S) = \sum A' B'$. Then
% \begin{equation*}
%     \Sigma(\mathcal S) = \sum_{\ell=1}^L A_\ell B_\ell + \sum_{\ell=1}^L A_\ell B_\ell' + A (B - B[Y^*, Z^*]) = \sum_{\ell=1}^L A_\ell B[Y_\ell, Z^*] + A(B - B[Y^*, Z^*]) 
% \end{equation*}
% Claim: $\sum_\ell A_\ell B[Y_\ell, Z^*] = A[X, Y^*] B[Y^*, Z^*]$. Fix $(x, z)$. For each $y \in Y$. Conclusion: Correct. But need better notation.}

\paragraph{Correctness of Property 2}
In both base cases we return partitions in which all parts $B_k$ are sparse, $\Ex[B_k] \leq 2^{-d}$. In the recursive case the output consists of the union of three different sets: For each $(X_\ell, Y_\ell, Z_\ell, A_\ell, B_\ell)$ the claim follows from Property~2 of \cref{lem:good-cube}, and for each element in $\mathcal S_\ell$ or $\mathcal S^*$ the claim holds by induction.

\paragraph{Correctness of Property 3}
For a collection $\mathcal S$ of tuples as returned by the algorithm, let us write
\begin{equation*}
    C(\mathcal S) = \sum_{(X', Y', Z', A', B') \in \mathcal S} |X'| \, |Y'| \, |Z'|.
\end{equation*}
Our goal is to prove that $C(\mathcal S) \leq 2^{h+1} \cdot (d + 2)^2 \cdot |X| \, |Y| \, |Z|$, where $\mathcal S$ is the output of our algorithm with inputs $(X, Y, Z, A, B, \epsilon, d, h)$. In the sparse case, if $\Ex[B] \leq 2^{-d}$, then we return the trivial decomposition with $C(\mathcal S) = |X| \, |Y| \, |Z|$. We prove by induction that otherwise the following bound applies:
\begin{equation*}
    C(\mathcal S) \leq 2^{h+1} \cdot (d + 2) (d + 1 - \log(\Ex[B]^{-1})) \cdot |X| \, |Y| \, |Z|
\end{equation*}
Clearly, this upper bound is true if $h = d$, so suppose that $h < d$. Then the algorithm recurses and we report $\mathcal S$ with $C(\mathcal S) = \sum_{\ell=1}^L |X_\ell| \, |Y_\ell| \, |Z_\ell| + \sum_{\ell=1}^L C(\mathcal S_\ell) + C(\mathcal S^*)$. In the following we bound these three contributions individually.

For the first contribution we readily exploit Property~4 of \cref{lem:good-cube}:
\begin{equation} \label{thm:2-path-decomposition:eq:1}
    \sum_{\ell=1}^L |X_\ell| \, |Y_\ell| \, |Z_\ell| \leq (d + 2) \cdot |X| \, |Y^*| \, |Z^*| \leq 2^h \cdot (d + 2) \cdot |X| \, |Y^*| \, |Z^*|.
\end{equation}
For the second contribution we exploit Property~5 of \cref{lem:good-cube}:
\begin{align}
    \sum_{\ell=1}^L C(\mathcal S_\ell)
    &\leq \sum_{\ell=1}^L 2^{h+1} \cdot (d + 2)^2 \cdot |X_\ell| \, |Y_\ell| \, |Z^* \setminus Z_\ell| \nonumber \\[.5ex]
    &\leq 2^{h+1} \cdot \gamma \cdot (d + 2)^3 \cdot |X| \, |Y^*| \, |Z^*| \nonumber \\[.5ex]
    &\leq 2^h \cdot (d + 2) \cdot |X| \, |Y^*| \, |Z^*|. \label{thm:2-path-decomposition:eq:2}
\end{align}
Here, in the last step, we have used our choice of \smash{$\gamma = \frac{1}{2(d + 2)^2}$}. For the third contribution, we distinguish two cases. If the density $\delta^*$ of the subgraph $B - B[Y^*, Z^*]$ is smaller than $2^{-d}$, then $C(\mathcal S^*) \leq |X| \, |Y| \, |Z|$. Otherwise, we have \smash{$\delta^* \leq \Ex[B] \cdot (1 - \frac{|Y^*| \, |Z^*|}{|Y| \, |Z|})$} (since $\Ex[B[Y^*, Z^*]] \geq \Ex[B]$ by Property~3 of \cref{lem:good-cube}) and therefore:
\begin{align}
    C(\mathcal S^*)
    &\leq 2^{h+1} \cdot (d + 2)(d + 1 - \log(\Ex[B]^{-1} \cdot (1 - \tfrac{|Y^*| \, |Z^*|}{|Y| \, |Z|})^{-1})) \cdot |X| \, |Y| \, |Z| \nonumber \\
    &\leq 2^{h+1} \cdot (d + 2)(d + 1 - \log(\Ex[B]^{-1}) + \log(1 - \tfrac{|Y^*| \, |Z^*|}{|Y| \, |Z|})) \cdot |X| \, |Y| \, |Z| \nonumber \\
    &\leq 2^{h+1} \cdot (d + 2)(d + 1 - \log(\Ex[B]^{-1}) - \tfrac{|Y^*| \, |Z^*|}{|Y| \, |Z|}) \cdot |X| \, |Y| \, |Z| \nonumber \\
    &\leq 2^{h+1} \cdot (d + 2)(d + 1 - \log(\Ex[B]^{-1})) \cdot |X| \, |Y| \, |Z| - 2^{h+1} \cdot (d + 2) \cdot |X| \, |Y^*| \, |Z^*|. \label{thm:2-path-decomposition:eq:3}
\end{align}
Summing over all three contributions~\eqref{thm:2-path-decomposition:eq:1},~\eqref{thm:2-path-decomposition:eq:2} and~\eqref{thm:2-path-decomposition:eq:3} yields the claimed bound on $C(\mathcal S)$.

\paragraph{Correctness of Property 4}
Let $M = \exp(d^4 \gamma^{-1} \poly(\epsilon^{-1}))$ be such that $|Y^*| \, |Z^*| \geq |Y| \, |Z| \,/\, M$ for the sets $Y^*, Z^*$ returned by \cref{lem:good-cube}. Moreover, let $L = \exp(d^3 \poly(\epsilon^{-1}))$ be as in \cref{lem:good-cube}. We prove by induction that
\begin{equation*}
    K = |\mathcal S| \leq 2^d M \cdot (2^d M L + L)^{d+1-h} \cdot \Ex[B],
\end{equation*}
where $\mathcal S$ is the set returned by the algorithm. This bound is easily verified in the two base cases. If the algorithm recurses instead then $|\mathcal S| \leq L + \sum_{\ell=1}^L |\mathcal S_\ell| + |\mathcal S^*|$. By induction we can bound
\begin{equation*}
    |\mathcal S_\ell| \leq 2^d M \cdot (2^d M L + L)^{d-h},
\end{equation*}
and
\begin{align*}
    |\mathcal S^*| &\leq 2^d M \cdot (2^d M L + L)^{d+1-h} \cdot \parens{\Ex[B] - \Ex[B[Y^*, Z^*]] \cdot \tfrac{|Y^*| \, |Z^*|}{|Y| \, |Z|}} \\
    &\leq 2^d M \cdot (2^d M L + L)^{d+1-h} \cdot \parens*{\Ex[B] - \Ex[B[Y^*, Z^*]] \cdot \tfrac{1}{M}} \\
    &\leq 2^d M \cdot (2^d M L + L)^{d+1-h} \cdot \parens*{\Ex[B] - \tfrac{1}{2^d M}},
\end{align*}
using in the last step that $\Ex[B[Y^*, Z^*]] \geq \Ex[B]$ by Property~3 of \cref{lem:good-cube}. Combining these bounds, it follows that
\begin{align*}
    |\mathcal S|
    &\leq L + L \cdot 2^d M \cdot (2^d M L + L)^{d-h} + 2^d M \cdot (2^d M L + L)^{d+1-h} \cdot (\Ex[B] - \tfrac{1}{2^d M}) \\
    &\leq (2^d M L + L)^{d+1-h} - (2^d M L + L)^{d+1-h} + 2^d M \cdot (2^d M L + L)^{d+1-h} \cdot \Ex[B] \\
    &= 2^d M \cdot (2^d M L + L)^{d+1-h} \cdot \Ex[B].
\end{align*}
In the end we plug in values for $L$ and $M$ to obtain $|S| \leq \exp(d^5 \gamma^{-1} \poly(\epsilon^{-1})) = \exp(d^7 \poly(\epsilon^{-1}))$, as stated.

\paragraph{Running Time}
From the previous consideration we also learn that the number of recursive calls is bounded by $\exp(d^7 \poly(\epsilon^{-1}))$. In each recursive call, the dominant step is to call $\RegCube$ in time $n^2 \cdot \exp(d^3 \poly(\epsilon^{-1}))$. Hence, the total time is $n^2 \cdot \exp(d^7 \poly(\epsilon^{-1}))$.
\end{proof}
% !TEX root = ../paper.tex
\section{Boolean Matrix Multiplication and Triangle Detection} \label{sec:triangle-detection}
In this section we formally derive our efficient algorithm for Boolean Matrix Multiplication from the regularity decompositions developed in the previous sections. Our algorithm relies on the following fine-grained reduction from BMM to Triangle Detection due to Vassilevska Williams and Williams~\cite{WilliamsW18}:

\begin{lemma}[Boolean Matrix Multiplication to Triangle Detection, \cite{WilliamsW18}] \label{lem:bmm-to-triangle}
If Triangle Detection is in time $\Order(n^3 / f(n))$ (for some nondecreasing function $f(n)$), then Boolean Matrix Multiplication is in time $\Order(n^3 / f(n^{1/3}))$.
\end{lemma}

\begin{theorem}[Triangle Detection] \label{thm:triangle-detection}
There is a deterministic combinatorial detecting whether a graph contains a triangle in time~\smash{$n^3 / 2^{\Omega(\sqrt[7]{\log n})}$}.
\end{theorem}
\begin{proof}
We assume without loss of generality that the input graph is tripartite, $(X, Y, Z, A, B, C)$. Let $\epsilon = \frac{1}{160}$ and let $d \geq 1$ be a parameter to be determined later. Using \cref{thm:2-path-decomposition} we decompose $(X, Y, Z, A, B)$ into pieces $\set{(X_k, Y_k, Z_k, A_k, B_k)}_{k=1}^K$. Further, write $C_k = C[X_k, Z_k]$. For each piece we distinguish two cases:
\begin{itemize}
    \item If $\Ex[A_k] \leq 2^{-d}$ or $\Ex[B_k] \leq 2^{-d}$ or $\Ex[C_k] \leq 2^{-\epsilon d / 2}$, then search in $(X_k, Y_k, Z_k, A_k, B_k, C_k)$ for a triangle. This step takes time $\Order(|X_k| \, |Y_k| \, |Z_k| \,/\, 2^{\epsilon d / 2})$ by exploiting the respective sparseness.
    \item Otherwise, simply return ``yes''.
\end{itemize}

For the correctness first observe that there is a triangle in the given graph if and only if there exists $(x, z) \in X \times Z$ with $(A B)(x, z) \geq 1$ and $C(x, z) = 1$. Since Property~1 of \cref{thm:2-path-decomposition} guarantees that~\makebox{$A B = \sum_k A_k B_k$}, there is a triangle in the original graph if and only if there is some $k \in [K]$ and $(x, z) \in X_k \times Z_k$ with $(A_k B_k)(x, z) \geq 1$ and $C_k(x, z) = 1$. The correctness of the first case is thus clear. But it remains to argue that if $\Ex[A_k] > 2^{-d}$ and $\Ex[B_k] > 2^{-d}$ and~\makebox{$\Ex[C_k] > 2^{-\epsilon d / 2}$}, then there exists a triangle in $(X_k, Y_k, Z_k, A_k, B_k, C_k)$. Indeed, by Property~2 of \cref{thm:2-path-decomposition} and the first two assumptions, we have that $A_k$ and $B_k^T$ are $(\epsilon, 2, d)$-regular and $\epsilon$-min-degree. In this case, \cref{thm:regular-product} implies that $A \circ B$ is $(\Ex[A] \Ex[B], 80\epsilon, 2^{-\epsilon d / 2})$-uniform. By definition, this means that not more than a $2^{-\epsilon d / 2}$-fraction of the entries in $A \circ B$ do not lie in the range $[(1 - 80\epsilon) \Ex[A] \Ex[B], (1 + 80\epsilon) \Ex[A] \Ex[B]] = [\frac{1}{2} \Ex[A] \Ex[B], \frac{3}{2} \Ex[A] \Ex[B]]$. In particular, (since we have~\makebox{$\Ex[A], \Ex[B] > 0$}) it follows that at most a $2^{-\epsilon d / 2}$-fraction of the entries in $A \circ B$ are nonzero. Using finally that $\Ex[C_k] > 2^{-\epsilon d / 2}$, we conclude that there exists some common entry~\makebox{$(x, z) \in X_k \times Z_k$} where both $(A B)(x, z) \geq 1$ and $C(x, z) = 1$.

Let us finally analyze the running time. Detecting a triangle in each sparse sub-instance takes time
\begin{equation*}
    \sum_{k=1}^K \frac{|X_k| \, |Y_k| \, |Z_k|}{2^{\Omega(d)}} \leq \frac{|X| \, |Y| \, |Z| \cdot 2(d + 2)^2}{2^{\Omega(d)}} = \frac{|X| \, |Y| \, |Z|}{2^{\Omega(d)}},
\end{equation*}
using Property~3 of \cref{thm:2-path-decomposition}. Furthermore, precomputating the regularity decomposition takes time $n^2 \cdot \exp(d^7 \poly(\epsilon^{-1})) = n^2 \cdot 2^{\Order(d^7)}$. This running is optimized by picking $d = \Theta(\sqrt[7]{\log n})$, where the constant is sufficiently small such that the preprocessing time becomes $\Order(n^{2.1})$, say. For this choice, the total running time is indeed~\smash{$n^3 / 2^{\Omega(d)} = n^3 / 2^{\Omega(\sqrt[7]{\log n})}$}.
\end{proof}

Our main \cref{thm:bmm} is immediate by combining \cref{lem:bmm-to-triangle,thm:triangle-detection}.
% !TEX root = ../paper.tex
\section{Triangle Enumeration} \label{sec:triangle-enumeration}
In this section we give an improved algorithm for enumerating triangles in graphs, based on our previous decomposition theorems:

\thmtriangleenum*

We make an important distinction: A \emph{triangle listing} algorithm receives as input a graph and returns as output a list of its $t$ triangles---here, we care about triangle listing algorithms with running times of the form $\Order(n^3 / f(n) + t)$. A \emph{triangle enumeration} algorithm first preprocesses a graph in time $\Order(n^3 / f(n))$. Afterwards, it can enumerate all triangles in the graph with constant delay (i.e., upon query, the algorithm spends time $\Order(1)$ to report the next triangle). For the majority of this section we work with triangle listing algorithms, but in \cref{sec:triangle-enumeration:sec:equiv} we show that both types are equivalent in our context.\footnote{We note that there has been work on developing triangle listing algorithms in time $\Order(n^3 / f(n) + t \cdot g(n))$, i.e., with a super-constant ``per-triangle'' cost $g(n)$. For this setting, algebraic fast matrix multiplication turns out to be useful and yields an algorithm in time~\smash{$\Order(n^{2.3716} + t^{0.4782} \cdot n^{1.5655})$}~\cite{BjorklundPWZ14}. Due to the super-constant per-triangle cost, however, this algorithm does not lead to nontrivial constant-delay enumeration algorithms.}

We structure the remainder of this section as follows: We quickly give the main idea of our algorithm in \cref{sec:triangle-enumeration:sec:overview}, with details following in \cref{sec:triangle-enumeration:sec:algo}. Finally, in \cref{sec:triangle-enumeration:sec:lower-bound} we include a proof that further improvements to our enumeration algorithm would entail a 3-SUM-algorithm that is faster than what is currently known.

\subsection{Triangle Listing---The Idea} \label{sec:triangle-enumeration:sec:overview}
In this short overview we give the main intuition behind our algorithm. We remark that all non-trivial algorithms for triangle enumeration are based on the Four-Russians technique~\cite{ArlazarovDKF70}. The following lemma states a stronger version for \emph{sparse} graphs that has implicitly appeared in~\cite{BansalW12,Chan15}:

\begin{lemma}[Four-Russians] \label{lem:four-russians}
Let $G = (X, Y, Z, A, B, C)$ be a tripartite graph. There is a deterministic algorithm that lists all $t$ triangles in $G$ in time
\begin{equation*}
    \Order\parens*{n^{2.3} + \frac{|X| \, |Y| \, |Z|}{(\log n)^{100}} + \frac{|X| \, |Y| \, |Z| \cdot \Ex[A \circ B] \cdot (\log\log n)^2}{(\log n)^2} + t}
\end{equation*}
(where $n \geq |X| + |Y| + |Z|$).
\end{lemma}

Without further assumptions on the graph we can only use the trivial bound $\Ex[A \circ B] \leq 1$, which recovers the original improvement of shaving two log-factors. Note that all subsequent works that score more log-shaves nevertheless still rely on this lemma.

To understand how we arrive at our improvement of nearly six log-shaves, and why six log-shaves appears to be a \emph{right} answer, suppose that the input graph is random (in the sense that it includes each edge independently with some probability~$\delta$). We claim that by trivial means, \cref{lem:four-russians} now yields an improvement by nearly six log-factors. There are two cases: If the graph is sparse,~\smash{$\delta \leq \frac{(\log\log n)^2}{(\log n)^2}$}, then
\begin{equation*}
    \Ex[A \circ B] \approx \Ex[A] \Ex[B] \approx \delta^2 \leq \frac{(\log\log n)^4}{(\log n)^4},
\end{equation*}
and the improvement is immediate. On the other hand, consider the dense case, \smash{$\delta \geq \frac{(\log\log n)^2}{(\log n)^2}$}. Intuitively, we exploit that whenever the graph is sufficiently dense, then the number of triangles dominates the running time in \cref{lem:four-russians}. Specifically, the number of triangles in the graph is~$t \approx \delta^3 \cdot n^3$ (assuming that the graph is random), and therefore the running time of \cref{lem:four-russians} is bounded by
\begin{equation*}
    \Order\parens*{\frac{n^3 \cdot \delta^2 \cdot (\log\log n)^2}{(\log n)^2} + t} = \Order\parens*{\frac{t \cdot (\log\log n)^2}{\delta \cdot (\log n)^2} + t} = \Order(t).
\end{equation*}

Of course, we cannot assume that the graph is perfectly random. Our natural approach is to apply our regularity decomposition to decompose the given graph into regular pieces that behave randomly (with respect to the quantity $\Ex[A \circ B]$ and the number of triangles $t$). The details are significantly more complicated though, as we cannot assume that the edges $C$ behave regularly.

\subsection{Triangle Listing---The Details} \label{sec:triangle-enumeration:sec:algo}
We now give the details of our algorithm, starting with a proof of the Four-Russians lemma. Recall that we work over a Word RAM model with word size $\Theta(\log n)$. In the proof of \cref{lem:four-russians} we specifically use that this model allows to construct a length-$n^{0.1}$ array that can be accessed \emph{in constant time} via keys of length $0.1 \log n$.

\begin{proof}[Proof of \cref{lem:four-russians}]
The goal is to give an algorithm that lists all triangles in a given tripartite graph $G = (X, Y, Z, A, B, C)$. We use the following notation: For a node $y \in Y$, let $N_A(y) \subseteq X$ and $N_B(y) \subseteq Z$ denote the sets of neighbors of~$y$, respectively. Let $s = \floor{\log n}^{100}, r = \floor{\frac{\log n}{1000 \log\log n}}$, and consider the following steps:
\begin{enumerate}
    \item We arbitrarily partition $X$ into $I = \ceil{|X| / s}$ groups $X_1, \dots, X_I$ of size at most $s$; similarly partition~$Z$ into $J = \ceil{|Z| / s}$ groups $Z_1, \dots, Z_J$ of size at most $s$.
    \item We enumerate each tuple $(i, j, S, T)$ where $i \in [I]$, $j \in [J]$ and where $S \subseteq X_i$, $T \subseteq Z_j$ have size $|S|, |T| \leq r$. Note that
    \begin{equation*}
        \log\binom{s}{r} \leq r \log s \leq \frac{\log n}{1000 \log\log n} \cdot 100 \log \log n = \frac{\log n}{10};
    \end{equation*}
    therefore, there are at most $n \cdot n \cdot n^{0.1} \cdot n^{0.1} = n^{2.2}$ such tuples $(i, j, S, T)$. Moreover, we can encode each such tuple in $\Order(\log n)$ bits which takes $\Order(1)$ machine words.\footnote{Of course, the exact bit representation matters. The easiest option here is to represent each set $S$ as a (sorted) list of its at most $r$ elements. As each element can be represented using $\log s$ bits, this representation indeed takes~\makebox{$r \log s \leq \frac{1}{10} \log n$} bits in total.} For each tuple~$(i, j, S, T)$ we prepare a list of all edges in $C[S, T]$ and store a pointer (of constant word size) to this list.
    \item Next, we enumerate each tuple $(y, i, j)$ where $y \in Y$, $i \in [I]$ and $j \in [J]$. We partition the set $N_A(y) \cap X_i$ (i.e., the set of neighbors of $y$ in $X_i$) arbitrarily into subsets $S$ of size $r$ (plus possibly one subset of size less than $r$); let $\mathcal S(y, i)$ denote the resulting partition. Similarly, we partition $N_B(y) \cap Z_j$ into subsets of size at most $r$ (plus possibly one subset of size less than~$r$); let~$\mathcal T(y, j)$ denote the resulting partition. Now, for each pair $S \in \mathcal S(y, i), T \in \mathcal T(y, j)$ we query list of edges associated to the tuple $(i, j, S, T)$. We enumerate each edge $(x, z)$ in this list associated to $(i, j, S, T)$ and store the triangle $(x, y, z)$.
\end{enumerate}
It is easy to verify that this algorithm lists all triangles in $G$. Let us focus on the running time. Step~1 runs in negligible time $\Order(n)$. In Step~2 we enumerate at most $n^{2.2}$ tuples $(i, j, S, T)$, and for each such tuple we spend time at most $\Order(s^2)$ to prepare the list of edges in $C[S, T]$. The total time of this step is $\Order(n^{2.2} \cdot s^2)$ which we loosely bound by $\Order(n^{2.3})$. In Step 3 we enumerate all~\makebox{$|Y| \cdot I \cdot J$} tuples $(y, i, j)$. For each such tuple, we spend time $\Order(s)$ to prepare the sets $\mathcal S(y, i)$ and $\mathcal T(y, j)$. Afterwards, we spend time $\Order(|\mathcal S(y, i)| \cdot |\mathcal T(y, j)|)$ plus the time to list all triangles. This listing cost is linear in $t$, so the running time of Step 3 is thus bounded by
\begin{align*}
    &\Order\parens*{\sum_{y \in Y} \sum_{\substack{i \in [I]\\j \in [J]}} \parens*{s + |\mathcal S(y, i)| \cdot |\mathcal T(y, j)|} + t} \\
    &\qquad= \Order\parens*{\sum_{y \in Y} \sum_{\substack{i \in [I]\\j \in [J]}} \parens*{s + \frac{|N_A(y) \cap X_i|}{r} \cdot \frac{|N_B(y) \cap Z_j|}{r}} + t} \\
    &\qquad= \Order\parens*{\sum_{y \in Y} \parens*{I J s + \frac{|N_A(y)| \cdot |N_B(y)|}{r^2}} + t} \\[.5ex]
    &\qquad= \Order\parens*{\frac{|X| \, |Y| \, |Z|}{s} + \frac{|X| \, |Y| \, |Z| \cdot \Ex[A \circ B]}{r^2} + t}.
\end{align*}
The time bound from the lemma statement follows by plugging in the chosen parameters $s$ and~$r$.
\end{proof}

\begin{theorem}[Triangle Listing] \label{thm:triangle-listing}
There is a deterministic algorithm that lists all $t$ triangles in a given graph in time $\Order(n^3 \,/\, (\log n)^6 \cdot (\log\log n)^{\Order(1)} + t)$.
\end{theorem}
\begin{proof}
Let $G = (X, Y, Z, A, B, C)$ be a given tripartite graph. We design a recursive algorithm that lists all triangles in $G$. To this end, we maintain one (global) list of triangles and each recursive call of the algorithm appends triangles to this list. Let $n$ denote the total number of nodes in the original graph $G$ (at the top level of the recursion), and let $\gamma, \delta, \epsilon \in (0, 1)$ and $d \geq 4 / \epsilon$ be parameters to be determined later.

The first step is to call \cref{thm:2-path-decomposition} on input $(X, Y, Z, A, B, \epsilon, d)$ to compute a decomposition~$\set{(X_k, Y_k, Z_k, A_k, B_k)}_{k=1}^K$. For each $k \in [K]$, we write for convenience~\makebox{$C_k = C[X_k, Z_k]$} and~\makebox{$G_k = (X_k, Y_k, Z_k, A_k, B_k, C_k)$}. By the guarantee of \cref{thm:2-path-decomposition} this decomposition partitions the set of triangles in $G$, and thus the remaining goal is to separately list all triangles in the graphs~$G_k$. To this end, for each $k \in [K]$, we distinguish the following cases:
\begin{enumerate}
    \item\emph{If $\Ex[A_k] \leq 2^{-\epsilon d / 4}$ or $\Ex[B_k] \leq 2^{-\epsilon d / 4}$:} List all triangles in $G_k$ in time $\Order(|X_k| \, |Y_k| \, |Z_k| \,/\, 2^{\epsilon d / 4})$.
    \item\emph{If $\Ex[A_k] \leq \delta$ and $\Ex[B_k] \leq \delta$:} List all triangles in $G_k$ by \cref{lem:four-russians}.
    \item\emph{If $\Ex[B_k] \geq \delta$:} We further subdivide the graph $G_k$ based on the approximate degrees in $X_k$ with respect to $Z_k$. Specifically, let $L = \ceil{\epsilon d / 2}$ and split~$X_k$ into buckets $X_{k, 1}, \dots, X_{k, L}$ defined by
    \begin{align*}
        X_{k, \ell} &= \set{x \in X_k : 2^{-\ell} < \deg_{C_k}(x) \leq 2^{-\ell+1}} \quad\text{(for $1 \leq \ell < L$)}, \\
        X_{k, L} &= \set{x \in X_k : \deg_{C_k}(x) \leq 2^{-L+1}}.
    \end{align*}
    For each $\ell \in [L]$, we define the matrix $C_{k, i} \in \set{0, 1}^{X_k \times Z_k}$ as the submatrix of $C_k$ obtained by zeroing out all rows not in $X_i$. (That is, $C_{k, i}[x, z] = C_k[x, z]$ if $x \in X_i$ and $C_k[x, z] = 0$ otherwise.) Writing $G_{k, \ell} = (X_k, Y_k, Z_k, A_k, B_k, C_{k, \ell})$, our remaining goal is to list the disjoint union of triangles in the graphs $G_{k, \ell}$. For each~\makebox{$\ell \in [L]$}, we instead distinguish the following three subcases:
    \begin{enumerate}
        \item[3.1]\emph{If $\Ex[C_{k, \ell}] \leq 2^{-L+1}$:} List all triangles in $G_{k, \ell}$ in time $\Order(|X_k| \, |Y_k| \, |Z_k| \,/\, 2^L)$.
        \item[3.2]\emph{If $|X_{k, \ell}| < \gamma |X_k|$:} Recurse on $(X_{k, \ell}, Y_k, Z_k, A_k[X_{k, \ell}, Y_k], B_k, C_k[X_{k, \ell}, Z_k])$. (It is easy to check that the triangles in this graph are exactly the triangles in $G_{k, \ell}$.)
        \item[3.3]\emph{If $\Ex[C_{k, \ell}] > 2^{-L+1}$ and $|X_{k, \ell}| \geq \gamma |X_k|$:} List all triangles in $(Y_k, X_k, Z_k, A_k^T, C_{k, \ell}, B_k)$ by \cref{lem:four-russians}. (It is easy to check that the triangles in this graph are in one-to-one correspondence to the triangles in $G_{k, \ell}$.) Note that we have exchanged the node and edge sets such that \cref{lem:four-russians} benefits from minimizing $\Ex[A_k^T \circ C_{k, \ell}]$.  %\shachar{is there a reason for this different order to parameters? I think it should be $(X_k, Y_k, Z_k, A_k, B_k, C_{k,\ell})$}
    \end{enumerate}
    \item If $\Ex[A_k] \geq \delta$: This case is symmetric to the previous case. More precisely, we can reduce to the previous case by considering instead the graph $(Z_k, Y_k, X_k, B_k^T, A_k^T, C_k^T)$ whose triangles are clearly in one-to-one correspondence with those in $G_k$.
\end{enumerate}
Finally, we add one more rule to the algorithm: As soon as we reach recursion depth $H$ (for some parameter $H$ to be determined), we simply solve the instance by brute-force in time $\Order(|X| \, |Y| \, |Z|)$. This completes the description of the algorithm. The correctness should be clear from the in-text explanations.

\paragraph{Running Time}
It remains to bound the running time. For simplicity, we already fix all the parameters here, and then analyze the cases individually:
\begin{align*}
    \epsilon &= \frac{1}{160}, \\
    d &= \ceil{64000 \log\log n}\vphantom{\frac11}, \\
    \gamma &= \frac{1}{8L(d + 2)^2} = \Theta\parens*{\frac{1}{(\log\log n)^3}}, \\
    \delta &= \frac{(\log\log n)^2}{\gamma (\log n)^2} = \Theta\parens*{\frac{(\log\log n)^5}{(\log n)^2}}, \\
    H &= \ceil*{\log \frac{(\log n)^6}{(\log\log n)^{14}}} = \Theta(\log\log n).
\end{align*}

\begin{claim}[Cases 1 and 3.1] \label{thm:triangle-listing:clm:case-1}
The total running time of Cases~1 and~3.1 is $\Order(|X| \, |Y| \, |Z| / (\log n)^{99})$.
\end{claim}
\begin{proof}
The algorithm deals with Cases 1 and 3.1 in time $\Order(|X_k| \, |Y_k| \, |Z_k| / 2^{\epsilon d / 4})$, which, by our choice of the parameters $\epsilon, d$, is $\Order(|X_k| \, |Y_k| \, |Z_k| / (\log n)^{100})$. In total, taking into account all~$k \in [K]$, and possibly the at most $L$ repetitions of Case 3.1, by \cref{thm:2-path-decomposition} this becomes
\begin{equation*}
    \sum_{k=1}^K \Order\parens*{\frac{L \cdot |X_k| \, |Y_k| \, |Z_k|}{(\log n)^{100}}} = \Order\parens*{\frac{d^2 L \cdot |X| \, |Y| \, |Z|}{(\log n)^{100}}} = \Order\parens*{\frac{|X| \, |Y| \, |Z|}{(\log n)^{99}}}. \qedhere
\end{equation*}
\end{proof}

\begin{claim}[Case 2] \label{thm:triangle-listing:clm:case-2}
The total running time of Case~2 is
\begin{equation*}
    \Order\parens*{K n^{2.3} + |X| \, |Y| \, |Z| \cdot \frac{(\log \log n)^{14}}{(\log n)^6} + t_{\text{\normalfont 2}}},
\end{equation*}
where $t_{\text{\normalfont 2}}$ is the number of triangles listed in Case 2.
\end{claim}
\begin{proof}
We can assume that $\Ex[A_k], \Ex[B_k] > 2^{-d}$ whenever we enter Case 2 (since Case 1 did not apply). Thus, the regularity decomposition (\cref{thm:2-path-decomposition}) guarantees that~$A_k$ and~$B_k^T$ are both $(\epsilon, 2, d)$-regular and $\epsilon$-min-degree. As a corollary of \cref{thm:regular-product} we obtain that
\begin{equation*}
    \Ex[A_k \circ B_k] \leq (1 + 80 \epsilon) \Ex[A] \Ex[B] + 2^{-\epsilon d / 2} \leq \tfrac{3}{2} \Ex[A] \Ex[B] + \Ex[A] \Ex[B] \leq \tfrac{5}{2} \delta^2.
\end{equation*}
Write $t_{G_k}$ to denote the number of triangles in $G_k$. By \cref{lem:four-russians} and \cref{thm:2-path-decomposition}, it follows that Case~2 indeed takes total time
\begin{align*}
    &\sum_{\substack{k \in [K]\\\text{case 2}}} \Order\parens*{n^{2.3} + |X_k| \, |Y_k| \, |Z_k| \cdot \parens*{\frac{1}{(\log n)^{100}} + \frac{\delta^2 (\log \log n)^2}{(\log n)^2}} + t_{G_k}} \\
    &\qquad= \sum_{\substack{k \in [K]\\\text{case 2}}} \Order\parens*{n^{2.3} + |X_k| \, |Y_k| \, |Z_k| \cdot \frac{(\log \log n)^{12}}{(\log n)^6} + t_{G_k}} \\
    &\qquad= \Order\parens*{K n^{2.3} + |X| \, |Y| \, |Z| \cdot \frac{d^2 (\log \log n)^{12}}{(\log n)^6} + t_{\text{2}}} \\
    &\qquad= \Order\parens*{K n^{2.3} + |X| \, |Y| \, |Z| \cdot \frac{(\log \log n)^{14}}{(\log n)^6} + t_{\text{2}}}. \qedhere
\end{align*}
\end{proof}

\begin{claim}[Case 3.3] \label{thm:triangle-listing:clm:case-3.3}
The total running time of Case 3.3 is
\begin{equation*}
    \Order\parens*{K L n^{2.3} + \frac{|X| \, |Y| \, |Z|}{(\log n)^{99}} + t_{\text{\normalfont 3.3}}},
\end{equation*}
where $t_{\text{\normalfont 3.3}}$ is the number of triangles listed in Case~3.3.
\end{claim}
\begin{proof}
Focus on some pair $(k, \ell)$ that falls into Case 3.3, i.e., where $\Ex[B_k] \geq \delta$,~$\Ex[C_{k, \ell}] > 2^{-L+1}$ and $|X_{k, \ell}| \geq \gamma |X_k|$. We must have $\ell < L$ (as otherwise $\Ex[C_{k, \ell}] \leq 2^{-L+1}$). The analysis of this case is inspired by the previous algorithm for purely random graphs. Our goal is to prove that both (i) the number of 2-paths in $G_{k, \ell}$ (via the edge parts $A_k$ and $C_{k, \ell}$) and (ii) the number of triangles in~$G_{k, \ell}$ behave as if~$G_{k, \ell}$ was purely random.

We start with (i). By the construction of $C_{k, \ell}$, it is immediate that the density of~$C_{k, \ell}$ is at least
\begin{equation*}
    \Ex[C_{k, \ell}] \geq \frac{|X_{k, \ell}|}{|X_k|} \cdot 2^{-\ell} \geq \gamma \cdot 2^{-\ell}.
\end{equation*}
Therefore, we can bound
\begin{align*}
    \Ex[A_k^T \circ C_{k, \ell}] &= \Ex_{x \in X_k} \deg_{A_k}(x) \cdot \deg_{C_{k, \ell}}(x) \\
    &\leq \Ex_{x \in X_k} \deg_{A_k}(x) \cdot 2^{-\ell+1} \\
    &\leq 2\gamma^{-1} \cdot \Ex[A_k] \cdot \Ex[C_{k, \ell}].
\end{align*}

Next, we turn to (ii) and bound the number of triangles in $G_{k, \ell}$ from below. Using \cref{thm:regular-product} we infer that at most a $2^{-\epsilon d/2} \leq 2^{-L}$-fraction of the entries in $\Ex[A_k \circ B_k]$ does not fall into the range $[(1 - 80\epsilon) \Ex[A_k] \Ex[B_k], (1 + 80\epsilon) \Ex[A_k] \Ex[B_k]] = [\frac12 \Ex[A_k] \Ex[B_k], \frac32 \Ex[A_k] \Ex[B_k]]$. Therefore, the number of triangles in $G_{k, \ell}$ is at least
\begin{align*}
    t_{G_{k, \ell}}
    &\geq |X_k| \, |Y_k| \, |Z_k| \cdot \tfrac{1}{2} \Ex[A_k] \cdot \Ex[B_k] \cdot (\Ex[C_{k, \ell}] - 2^{-L}) \\
    &\geq |X_k| \, |Y_k| \, |Z_k| \cdot \tfrac{1}{4} \cdot \Ex[A_k] \cdot \Ex[B_k] \cdot \Ex[C_{k, \ell}] \\
    &\geq |X_k| \, |Y_k| \, |Z_k| \cdot \tfrac{\gamma}{8} \cdot \Ex[B_k] \cdot \Ex[A_k^T \circ C_{k, \ell}] \\
    &\geq |X_k| \, |Y_k| \, |Z_k| \cdot \tfrac{\gamma \delta}{8} \cdot \Ex[A_k^T \circ C_{k, \ell}].
\end{align*}

By combining both statements, we can bound the running time of \cref{lem:four-russians} as follows:
\begin{align*}
    &\sum_{\substack{k \in [K], \ell \in [L]\\\text{case 3.3}}} \Order\parens*{n^{2.3} + |X_k| \, |Y_k| \, |Z_k| \cdot \parens*{\frac{1}{(\log n)^{100}} + \frac{\Ex[A_k^T \circ C_{k, \ell}] \cdot (\log \log n)^2}{(\log n)^2}} + t_{G_{k, \ell}}} \\
    &\qquad\leq \sum_{\substack{k \in [K], \ell \in [L]\\\text{case 3.3}}} \Order\parens*{n^{2.3} + \frac{|X_k| \, |Y_k| \, |Z_k|}{(\log n)^{100}} + t_{G_{k, \ell}} \cdot \parens*{\frac{(\log\log n)^2}{(\gamma \delta/8) (\log n)^2} + 1}} \\
    &\qquad\leq \sum_{\substack{k \in [K], \ell \in [L]\\\text{case 3.3}}} \Order\parens*{n^{2.3} + \frac{|X_k| \, |Y_k| \, |Z_k|}{(\log n)^{100}} + t_{G_{k, \ell}}} \\
    &\qquad= \Order\parens*{K L n^{2.3} + \frac{d^2 L \cdot |X| \, |Y| \, |Z|}{(\log n)^{100}} + t_{\text{3.3}}} \\
    &\qquad= \Order\parens*{K L n^{2.3} + \frac{|X| \, |Y| \, |Z|}{(\log n)^{99}} + t_{\text{3.3}}}. \qedhere
\end{align*}
\end{proof}

It only remains to analyze Case~3.2, which involves recursive calls to our algorithm. Let us write $T(|X|, |Y|, |Z|, t, h)$ to express the total running time of our algorithm, given an input graph with vertex parts $X, Y, Z$ and with at most $t$ triangles, at recursion depth $0 \leq h \leq H$. Then finally:

\begin{claim}[Total Running Time] \label{thm:triangle-listing:clm:total}
The total running time is bounded by
\begin{equation*}
    T(|X|, |Y|, |Z|, t, h) \leq c \cdot \parens*{(2 K L)^{H-h} \cdot n^{2.3} + |X| \, |Y| \, |Z| \cdot \frac{(\log\log n)^{14} \cdot 2^{h+1}}{(\log n)^6} + t},
\end{equation*}
for some sufficiently large constant $c$.
\end{claim}
\begin{proof}
In the base case we have $T(|X|, |Y|, |Z|, t, H) = \Order(|X| \, |Y| \, |Z|)$ which indeed satisfied the claim (as $2^H \geq (\log n)^6 / (\log\log n)^{14}$). So focus on the case that $h < H$. Then, given the previous \cref{thm:triangle-listing:clm:case-1,thm:triangle-listing:clm:case-2,thm:triangle-listing:clm:case-3.3}, the running time of the Cases~1, 2,~3.1 and~3.3 is bounded by
\begin{equation} \label{thm:triangle-listing:clm:total:eq:others}
    c \cdot \parens*{K L n^{2.3} + |X| \, |Y| \, |Z| \cdot \frac{(\log \log n)^{14}}{(\log n)^6} + (t - t_{\text{3.2}})},
\end{equation}
where $c$ is some sufficiently large constant, and where $t_{\text{3.2}}$ is the number of triangles listed in Case~3.2 (such that $t - t_{\text{3.2}}$ is the number of triangles listed in all other cases). By induction we can bound the contribution of Case~3.2 as follows; recall that $|X_{k, \ell}| < \gamma |X_k|$ for any pair $(k, \ell)$ falling into Case~3.2:
\begin{align}
    &\sum_{\substack{k \in [K], \ell \in [L]\\\text{case 3.2}}} T(|X_{k, \ell}|, |Y_k|, |Z_k|, t_{G_{k, \ell}}, h + 1) \nonumber\\
    &\qquad\leq \sum_{\substack{k \in [K], \ell \in [L]\\\text{case 3.2}}} c \cdot \parens*{(2 K L)^{H-h-1} \cdot n^{2.3} + |X_{k, \ell}| \, |Y_k| \, |Z_k| \cdot \frac{(\log\log n)^{14} \cdot 2^{h+2}}{(\log n)^6} + t_{G_{k, \ell}}} \nonumber\\
    &\qquad\leq \sum_{\substack{k \in [K], \ell \in [L]\\\text{case 3.2}}} c \cdot \parens*{(2 K L)^{H-h-1} \cdot n^{2.3} + \gamma \cdot |X_k| \, |Y_k| \, |Z_k| \cdot \frac{(\log\log n)^{14} \cdot 2^{h+2}}{(\log n)^6} + t_{G_{k, \ell}}} \nonumber\\
    &\qquad c \cdot \parens*{K L (2 K L)^{H - h - 1} \cdot n^{2.3} + 2\gamma L (d + 2)^2 \cdot |X| \, |Y| \, |Z| \cdot \frac{(\log\log n)^{14} \cdot 2^{h+2}}{(\log n)^6} + t_{\text{3.2}}}. \label{thm:triangle-listing:clm:total:eq:recursive}
\end{align}
By our choice of parameters, $2 \gamma L (d + 2)^2 = \frac14$. Therefore, the total running time (which is obtained as the sum of \cref{thm:triangle-listing:clm:total:eq:others,thm:triangle-listing:clm:total:eq:recursive}) is at most
\begin{align*}
    &T(|X|, |Y|, |Z|, t, h) \\
    &\qquad\leq c \cdot \parens*{2 K L (2 K L)^{H - h - 1} \cdot n^{2.3} + |X| \, |Y| \, |Z| \cdot \frac{(\log \log n)^{14} \cdot 2^{h+1}}{(\log n)^6} + (t - t_{\text{3.2}} + t_{\text{3.2}})} \\
    &\qquad= c \cdot \parens*{(2 K L)^{H - h} \cdot n^{2.3} + |X| \, |Y| \, |Z| \cdot \frac{(\log \log n)^{14} \cdot 2^{h+1}}{(\log n)^6} + t},
\end{align*}
as stated.
\end{proof}

To obtain the time bound claimed in the theorem statement, recall that we initially call the algorithm at recursion depth $h = 0$, that $K \leq \exp(\poly(d, \epsilon^{-1})) = \exp(\poly(\log\log n))$, and that~\makebox{$H = \Order(\log\log n)$}. Hence, the term $(2 K L)^H \cdot n^{2.3} \leq n^{2.3+\order(1)}$ is negligible in the total running time.
\end{proof}

\subsection{Triangle Enumeration versus Triangle Listing} \label{sec:triangle-enumeration:sec:equiv}
We finally turn our triangle listing algorithm into an enumeration algorithm. In fact, we prove the following equivalence:\footnote{This equivalence can easily be extended to an equivalence for arbitrary subgraphs $H$. Specifically, if $H$ has $k$ vertices, then listing $H$-subgraphs in time $\Order(n^k / f(n) + t)$ is equivalent to enumerating $H$-subgraphs with constant delay in preprocessing time $\Order(n^k / f(n))$. However, this equivalence has only limited use, as for many graphs $H$ much faster listing and enumeration algorithms are known. For instance, 4-cycles can be listed in time $\Order(n^2 + t)$ by a simple folklore algorithm.}

\begin{lemma}[Equivalence of Triangle Enumeration and Listing] \label{lem:equivalence-enumerate}
The following equivalences hold (in terms of deterministic algorithms):
\begin{enumerate}
    \item If constant-delay triangle enumeration is possible with preprocessing time $\Order(n^3 / f(n))$ (for some function $f(n)$), then triangle listing is in time $\Order(n^3 / f(n) + t)$.
    \item If triangle listing is in time $\Order(n^3 / f(n) + t)$ (for some computable function $f(n) = n^{\order(1)}$), then constant-delay triangle enumeration is possible with preprocessing time $\Order(n^3 / f(n^{1/2}))$.
\end{enumerate}
\end{lemma}

To obtain the equivalence in terms of \emph{deterministic} algorithms we rely on the following combinatorial algorithm by Fox, Lovász and Zhao to approximately count triangles (or in fact, arbitrary subgraphs):

\begin{lemma}[Approximate Triangle Counting~\cite{FoxLZ19}] \label{lem:approx-triangles}
Let $\epsilon > 0$. There is a deterministic algorithm approximating the number of triangles in a graph up to an additive error of~$\epsilon n^3$ in time $n^2 \cdot \poly(\epsilon^{-1})$.
\end{lemma}

\begin{proof}[Proof of \cref{lem:equivalence-enumerate}]
The first item is trivial: Simply preprocess the graph in time $\Order(n^3 / f(n))$ and then enumerate all triangles in time $\Order(t)$. We thus focus on the second item, and design a triangle enumeration algorithm with constant delay.

\paragraph{Preprocessing}
Assume without loss of generality that $G$ is tripartite with vertex sets $X, Y, Z$ of size $n$ each. We partition each vertex part into $g = \ceil{n^{1/2}}$ groups $X = X_1 \uplus \dots \uplus X_g$, $Y = Y_1 \uplus \dots \uplus Y_g$ and~\makebox{$Z = Z_1 \uplus \dots \uplus Z_g$} of size at most $\ceil{n^{1/2}}$. For each triple $i, j, k \in [g]$, let $G_{i, j, k}$ denote the subgraph induced by $X_i \cup Y_j \cup Z_k$. Clearly the triangles in~$G$ are perfectly partitioned into the triangles in $(G_{i, j, k})_{i, j, k}$. Let $t_{i, j, k}$ denote the number of triangles in $G_{i, j, k}$.

Let $\epsilon = \frac14 / f(n^{1/2})$. For each triple $i, j, k \in [g]$, we compute an approximation $\widetilde t_{i, j, k}$ of the number of triangles in $G_{i, j, k}$ with absolute error~\makebox{$\epsilon n^{3/2} = \frac14 n^{3/2} / f(n^{1/2})$}. We call a triple $i, j, k$ with~\makebox{$\widetilde t_{i, j, k} \leq n^{3/2} / f(n^{1/2})$} \emph{light} and \emph{heavy} otherwise. For each light triple~$i, j, k$ we list all triangles in $G_{i, j, k}$ using the efficient listing algorithm and discard $G_{i, j, k}$ afterwards. We store all triangles discovered in this step in one global list to be enumerated later. We sort the remaining heavy triples in descending order according to their approximate triangle counts $\widetilde t_{i, j, k}$; let $G_1, \dots, G_r$ denote the heavy graphs in the resulting order. As the final step of the preprocessing phase, we list all triangles in $G_1$ in brute-force time $\Order(n^{3/2})$.

\paragraph{Preprocessing Time}
Before proceeding to the description of the enumeration phase, we quickly analyze the running time of the preprocessing phase. Computing the approximate triangle counts for all subgraphs takes time $\Order(n^{3/2} \cdot n \cdot \poly(f(n^{1/2}))) = n^{5/2 + \order(1)}$. Observe that the number of triangles in each light subgraph is at most $n^{3/2} / f(n^{1/2}) + \epsilon n^{3/2} = \Order(n^{3/2} / f(n^{1/2}))$. Therefore, applying the listing algorithm to all light subgraphs takes time $\Order(n^{3/2} \cdot n^{3/2} / f(n^{1/2})) = \Order(n^3 / f(n^{1/2}))$. Finally, listing all triangles in~$G_1$ takes negligible time.

\paragraph{Enumeration}
In the enumeration phase we first enumerate all triangles from the light subgraphs. The more interesting part is to enumerate the triangles involving the heavy subgraphs. Here we proceed as follows: We maintain an \emph{active} heavy subgraph $i \in [r]$ for which we have prepared a list of all of its triangles. Initially we set $i = 1$---and indeed, we have prepared a list of all triangles in $G_1$. The idea is that while we enumerate triangles from the active graph $G_i$, we simultaneously prepare a list of all triangle in the next graph $G_{i+1}$. To this end, we run the efficient listing algorithm on $G_{i+1}$, and with each triangle from $G_i$ that we list we advance the algorithm by~$\Order(1)$ computation steps. To prove that this approach succeeds (i.e., that we have completed the execution of the listing algorithm on~$G_{i+1}$ when all triangles from~$G_i$ have been listed), we show that $t_i \geq \Omega(n^{3/2} / f(n^{1/2}) + t_{i+1})$, where we write $t_i$ to denote the number of triangles in~$G_i$. Indeed, we have that $t_i \geq t_{i+1} - 2\epsilon n^{3/2}$ and $t_i \geq n^{3/2} / f(n^{1/2}) - \epsilon n^{3/2}$ (as~$G_i$ is heavy), and thus~\makebox{$t_i \geq \Omega(n^{3/2} / f(n^{1/2}) + t_{i+1})$}. After we have enumerated all triangles from~$G_i$, we discard~$G_i$ and consider~$G_{i+1}$ the next active subgraph. Note that this algorithm indeed lists all triangles with constant delay.
\end{proof}

The proof of \cref{thm:triangle-enumeration} is complete by combining \cref{thm:triangle-listing,lem:equivalence-enumerate}.

\subsection{Conditional Optimality} \label{sec:triangle-enumeration:sec:lower-bound}
The famous 3-SUM problem is to test whether in a given set $S$ of $n$ integers there is a solution to the equation $a + b + c = 0$. This problem can be solved naively in time $\Order(n^2)$ and the fastest known algorithm, due to Baran, Demaine and Pătraşcu, runs in expected time $\Order(n^2 / (\log n)^2 \cdot (\log\log n)^2)$.

From previous work on the 3-SUM-hardness of triangle listing/enumeration by Pătraşcu~\cite{Patrascu10} and by Kopelowitz, Pettie and Porat~\cite{KopelowitzPP16}, it follows that our algorithm is \emph{conditionally optimal} in the following (weak) sense:  An algorithm that enumerates all triangles in a graph with preprocessing time~\makebox{$\Order(n^3 / (\log n)^{6+\epsilon})$} and constant delay entails an algorithm for the 3-SUM problem in expected time~\makebox{$\Order(n^2 / (\log n)^{2+\epsilon'})$} (i.e., a $(\log n)^{\epsilon'}$ improvement over the Baran-Demaine-Pătraşcu algorithm). Since this statement is not explicit in~\cite{KopelowitzPP16}, we devote this section to an almost-self-contained proof. Specifically, we prove the following statement:

\begin{lemma}[Reducing 3-SUM to Triangle Listing] \label{lem:3-SUM-to-triangle-listing}
Let $\alpha \geq 0$. If there is a randomized algorithm listing all $t$ triangles in a graph in expected time $\Order(n^3 / f(n) + t)$ (for some computable nondecreasing function $f(n)$), then there is a randomized 3-SUM algorithm in expected time~$\Order(n^2 / f(n^{2/3})^{1/3})$.\footnote{We remark that while this reduction is in terms of \emph{expected} running times, the statement can easily be strengthened such that the running time of the constructed 3-SUM algorithm holds with high probability (e.g., by first applying the known self-reduction for 3-SUM~\cite{GronlundP18,LincolnWWW16} and splitting into polynomially many subinstances).}
\end{lemma}

The proof relies on the following standard lemma on linear hashing:

\begin{lemma}[Linear Hashing~\cite{CormenLRS09}] \label{lem:linear-hashing}
Let $n \geq m \geq 1$. There is a family $\mathcal H = \set{h : [-n\,.\,.\,n] \to [m]}$ of hash functions that can be sampled in expected time $\poly(\log n)$ and evaluated in constant time, and that satisfies the following properties:
\begin{itemize}
    \item (Almost-Linearity) There exists some constant-size set $\Phi$ such that for all $h \in \mathcal H$ and all keys~\makebox{$a, b \in [-n\,.\,.\,n]$}, we have $h(a + b) - h(a) - h(b) + h(0) \in \Phi$.
    \item (Pairwise Independence) For all distinct keys $a, b \in [-n\,.\,.\,n]$ and buckets $x, y \in [m]$:\newline$\Pr_{h \in \mathcal H}[\text{$h(a) = x$ and $h(b) = y$}] \leq \Order(\frac{1}{m^2})$.% \shachar{should this be $O(1/m^2)$?}
\end{itemize}
\end{lemma}

\begin{proof}[Proof of \cref{lem:3-SUM-to-triangle-listing}]
Let $S$ denote the given 3-SUM-instance, and assume that $S \subseteq [-n^c\,.\,.\,n^c]$ for some constant~$c$. As a first step, we randomly sample $\Order(n \log n)$ pairs $a, b \in S$ and test whether~\makebox{$-a - b \in S$}. If we find a solution in this step, we stop and return ``yes''.

Otherwise, sample three linear hash functions $h_1, h_2, h_3 : [-n^c\,.\,.\,n^c] \to [m]$ as in the previous lemma, and construct the following tripartite graph $G = (X, Y, Z, A, B, C)$. As vertex parts, we take
\begin{align*}
    X &= [m] \times [m] \times \set{h_3(0)}, \\
    Y &= [m] \times \set{h_2(0)} \times [m], \\
    Z &= \set{h_1(0)} \times [m] \times [m].
\end{align*}
We add an edge $(x, y) \in X \times Y$ to $A$ if and only if there is some $a \in S$ such that
\begin{align*}
    y_1 - x_1 - h_1(a) + h_1(0) \in \Phi, \\
    y_2 - x_2 - h_2(a) + h_2(0) \in \Phi, \\
    y_3 - x_3 - h_3(a) + h_3(0) \in \Phi.
\end{align*}
In this case, we say that the edge $(x, y)$ is \emph{labeled} with $a$. We add edges to $B$ and $C$ in the analogous way. We then use the efficient triangle listing algorithm to list all triangles in $G$ (viewing $G$ as an unlabeled graph). For each triangle $(x, y, z)$ that is returned, test whether it has three edge labels~$(a, b, c)$ that satisfy $a + b + c = 0$. In this case we report ``yes'', and if no such triangle is found, return ``no''.

\paragraph{Correctness}
For the correctness it suffices to show that, if there exist $a, b, c \in S$ with $a + b + c = 0$, then there is a triangle in~$G$ with edge labels $(a, b, c)$. To see this, consider the triple $(x, y, z)$ where
\begin{align*}
    x = (h_1(c), h_2(-a), h_3(0)), \\
    y = (h_1(-b), h_2(0), h_3(a)), \\
    z = (h_1(0), h_2(b), h_3(-c)).
\end{align*}
We claim that $(x, y, z)$ forms a triangle. This involves testing nine constraints. For the sake of brevity, we only prove the first such constraint, namely that $y_1 - x_1 - h_1(a) + h_1(0) \in \Phi$. And indeed, given that $y_1 - x_1 - h_1(a) = h_1(-b) - h_1(c) - h_1(a) = h_1(c + a) - h_1(c) - h_1(a)$, the constraint follows by the almost-linearity property of \cref{lem:linear-hashing}.

\paragraph{Running Time}
Sampling solutions and constructing the graph $G$ takes negligible time, so the critical contribution is the running time of the listing algorithm. To this end, we first bound the expected number of triangles in the graph.

\begin{claim*}
The expected number of triangles in $G$ is~\smash{$\Order(n + \frac{n^3}{m^3})$}.
\end{claim*}
\begin{proof}
Fix a triple $a, b, c \in S$. We analyze how many triangles $(x, y, z)$ are labeled with~$(a, b, c)$. Recall that for $(x, y, z)$ to be a triangle, it satisfies nine constraints---three for each edge. Focus on the constraints involving $h_1$:
\begin{align*}
    y_1 - x_1 - h_1(a) + h_1(0) &\in \Phi, \\
    z_1 - y_1 - h_1(b) + h_1(0) &\in \Phi, \\
    x_1 - z_1 - h_1(c) + h_1(0) &\in \Phi.
\end{align*}
Recall that necessarily $z_1 = h_1(0)$. Thus, the second and third constraints imply that there are only $|\Phi| = \Order(1)$ feasible choices for $x_1$ and $y_1$. In particular, it follows that there are at most $\Order(1)$ triangles labeled with $(a, b, c)$. Next, let us write $k \Phi = \set{\phi_1 + \dots + \phi_k : \phi_1, \dots, \phi_k \in \Phi}$. Summing all three constraints, we obtain that $3h_1(0) - h_1(a) - h_1(b) - h_3(c) \in 3\Phi$, and by applying the almost-linearity twice, it follows that $h_1(0) - h_1(a + b + c) \in 5\Phi$. By the analogous argument for~$h_2$ and $h_3$, we conclude that there is a triangle labeled with $(a, b, c)$ only if
\begin{align*}
    h_1(0) - h_1(a + b + c) &\in 5 \Phi, \\
    h_2(0) - h_2(a + b + c) &\in 5 \Phi, \\
    h_3(0) - h_3(a + b + c) &\in 5 \Phi.
\end{align*}

Finally, to bound the expected number of triangles, we distinguish two cases. On the one hand, the number of triples with $a + b + c = 0$ is at most $\Order(n)$ as otherwise, with high probability, we would have detected a 3-SUM solution in the first step of the algorithm. Each such triple contributes at most $\Order(1)$ triangles. On the other hand, there are up to $n^3$ triples with $a + b + c \neq 0$. Each such triple contributes at most $\Order(1)$ triangles, and only if the final three constraints are satisfied. By pairwise independence, each constraint holds with probability at most~\smash{$\Order(\frac{|\Phi|}{m}) = \Order(\frac{1}{m})$}, and the three constraints are independent. It follows that these triples contribute \smash{$\Order(\frac{n^3}{m^3})$} triangles in expectation.
\end{proof}

Observe that the graph $G$ has $\Order(m^2)$ vertices. Therefore, given the previous claim, the total expected running time is bounded by $\Order(m^6 / f(m^2) + n^3 / m^3)$. By choosing $m = \ceil{n^{1/3} \cdot f(n^{2/3})^{1/9}}$, this becomes $\Order(n^2 / f(n^{2/3})^{1/3})$ (using that $f$ is nondecreasing).
\end{proof}

\bibliographystyle{plainurl}
\bibliography{references}

\end{document}